\documentclass{article}
\usepackage[utf8]{inputenc}
\usepackage[T1]{fontenc}
\usepackage{lineno,hyperref}

\usepackage[auth-sc]{authblk}

\usepackage{verbatim}
\usepackage{graphicx}
\usepackage{rotating}

\usepackage{url}
\usepackage{xcolor}
\definecolor{linkColor}{RGB}{156,78,13}
\usepackage{hyperref}
\hypersetup{
    colorlinks = true,
    allcolors = linkColor 
    }
\usepackage{xspace}
\usepackage{pslatex}

\usepackage{datetime}
\newdateformat{versiondate}{%
\THEMONTH\THEDAY}

\usepackage{pgf,tikz,environ}
\usetikzlibrary{arrows}
\usetikzlibrary{cd}
\usetikzlibrary{positioning}
\usetikzlibrary{fadings}
\usetikzlibrary{decorations.pathreplacing}
\usetikzlibrary{decorations.pathmorphing}
\tikzset{snake it/.style={decorate, decoration=snake}}

\usepackage{amsthm}
\usepackage[fleqn]{amsmath} 
\usepackage{amsfonts}
\usepackage{amssymb}
\usepackage{stix}

\usepackage{mathtools}

\usepackage{sectsty}
\allsectionsfont{\fontfamily{lmss}\selectfont}

\definecolor{parchment}{RGB}{214, 204, 169}
\theoremstyle{plain}
  
\newtheorem{theorem}{Theorem}[section]
\newtheorem{corollary}[theorem]{Corollary}
\newtheorem{lemma}[theorem]{Lemma}

\theoremstyle{definition}
\newtheorem{definition}[theorem]{Definition}
\newtheorem{claim}[theorem]{Claim} 

\usepackage{mdframed} 
\surroundwithmdframed[hidealllines=true,backgroundcolor=parchment!40]{theorem}
\surroundwithmdframed[hidealllines=true,backgroundcolor=parchment!40]{corollary}
\surroundwithmdframed[hidealllines=true,backgroundcolor=parchment!40]{lemma}
\surroundwithmdframed[hidealllines=true,backgroundcolor=parchment!40]{problem}
\surroundwithmdframed[hidealllines=true,backgroundcolor=parchment!40]{example}
\surroundwithmdframed[hidealllines=true,backgroundcolor=parchment!40]{proposition}
\surroundwithmdframed[hidealllines=true,backgroundcolor=parchment!40]{conjecture}
\surroundwithmdframed[hidealllines=true,backgroundcolor=parchment!40]{definition}
\surroundwithmdframed[hidealllines=true,backgroundcolor=parchment!40]{claim}

\usepackage{enumitem}
\usepackage{amsmath}
\usepackage{framed} 

\newcommand{\typesetoperator}[1]{\mathbf{#1}}

\DeclareMathOperator{\imw}{\typesetoperator{imw}}
\DeclareMathOperator{\vimw}{\typesetoperator{vimw}}
\DeclareMathOperator{\reach}{\typesetoperator{reach}}
\DeclareMathOperator{\cost}{\typesetoperator{cost}}

\newcommand{\typesetproblem}[1]{\textsc{#1}}
\DeclareMathOperator{\starexp}{\typesetproblem{StarExp}}
\DeclareMathOperator{\tempeuler}{\typesetproblem{TempEuler}}
\DeclareMathOperator{\minreachdelete}{\typesetproblem{MinReachDelete}}
\newcommand{\col}{$3$-\typesetproblem{Coloring}\xspace}
\DeclareMathOperator{\p}{\typesetproblem{P}}
\DeclareMathOperator{\np}{\typesetproblem{NP}}
\DeclareMathOperator{\fpt}{\typesetproblem{FPT}}

\title{\fontfamily{lmss}\selectfont  Edge exploration of temporal graphs.}
\author{Benjamin Merlin Bumpus\thanks{Supported by an EPSRC doctoral training account.} $^{}$ \thanks{Funded by the European Research Council (ERC) under the European Union's Horizon 2020 research and innovation programme (grant agreement No 803421, ReduceSearch).} ~\& Kitty Meeks\thanks{Supported by a Royal Society of Edinburgh Personal Research Fellowship, funded by the Scottish Government, and EPSRC grant EP/T004878/1.}}
\date{Last compilation: \today}

\begin{document}

\maketitle

\begin{abstract}
We introduce a natural temporal analogue of Eulerian circuits and prove that, in contrast with the static case, it is $\np$-hard to determine whether a given temporal graph is temporally Eulerian even if strong restrictions are placed on the structure of the underlying graph and each edge is active at only three times. However, we do obtain an $\fpt$-algorithm with respect to a new parameter called \emph{interval-membership-width} which restricts the times assigned to different edges. Our techniques also allow us to resolve two open questions of Akrida, Mertzios and Spirakis [CIAC 2019] concerning a related problem of exploring temporal stars. Furthermore, we introduce a vertex-variant of interval-membership-width (which can be arbitrarily larger than its edge-counterpart) and use it to obtain an $\fpt$-time algorithm for a natural \emph{vertex}-exploration problem that remains hard even when interval-membership-width is bounded.
\end{abstract}

\section{Introduction}\label{sec_edge_expl:intro}
Many real-world problems can be formulated and modeled in the language of graph theory. However, real-world networks are often not \emph{static}. They change over time and their edges may appear or disappear (for instance friendships may change over time in a social network). Such networks are called \emph{dynamic} or \emph{evolving} or \emph{temporal} and their structural and algorithmic properties have been the subject of active study in recent years \cite{Akrida2017,ArnaudtempSurvey,Himmel2017,HOLME201297,OthonTempSurvey}.
Some of the most natural and most studied topics in the theory of temporal graphs are: temporal walks (in which consecutive edges appear at increasing times), temporal paths and the corresponding notions of temporal reachability \cite{AKRIDA202065,axiotis2016size,temp_conn_cpts,temporalFeedbackVertexSet,mertzios2019temporal,KEMPE2002820,temp_path_comput,xuan2003computing}. Related to these notions is the study of explorability of a temporal graph which asks whether it is possible to visit all vertices or edges of a temporal graph via some temporal walk. 

Temporal vertex-exploration problems (such as temporal variants of the Travelling Salesman problem) have already been thoroughly studied \cite{starexp,OnTemporalGraphExploration,OthonTSP}. In contrast, here we focus on temporal \emph{edge}-exploration and specifically we study \emph{temporally Eulerian graphs}. Informally, these are temporal graphs admitting a temporal circuit that visits every edge at exactly one time (i.e. a temporal circuit that yields an Euler circuit in the underlying static graph). 

Deciding whether a static graph is Eulerian is a prototypical example of a polynomial time solvable problem. In fact this follows from Euler's characterization of Eulerian graphs dating back to the 18$^{\text{th}}$ century \cite{eulerBridges}. In contrast, here we show that, unless $\p = \np$, a characterization of this kind cannot exist for temporal graphs. In particular we show that deciding whether a temporal graph is \emph{temporally Eulerian} is $\np$-complete even if strong restrictions are placed on the structure of the underlying graph and each edge is active at only three times.

The existence of problems that are tractable on static graphs, but $\np$-complete on temporal graphs is well-known \cite{starexp,ArnaudtempSurvey,mertzios2019computing,OthonTempSurvey}. In fact there are examples of problems whose temporal analogues remain hard even on trees \cite{starexp,mertzios2019computing}. Thus the need for parameters that take into account the temporal structure of the input is clear. Some measures of this kind (such as temporal variants of feedback vertex number~\cite{temporalFeedbackVertexSet} and tree-width~\cite{temp_tw}) have already been studied. Unfortunately we shall see that these parameters will be of no use to us since the problems we consider here remain $\np$-complete even when these measures are bounded by constants on the underlying static graph. To overcome these difficulties, we introduce a new purely-temporal parameter called \emph{interval-membership-width}. Parameterizing by this measure we find that the problem of determining whether a temporal graph is temporally Eulerian is in $\fpt$. 

Temporal graphs of low interval-membership-width are `temporally sparse' in the sense that only few edges are allowed to appear both before and after any given time. We point out that this parameter does \emph{not} depend on the structure of the underlying static graph, but it is instead influenced only by the temporal structure. We believe that interval-membership-width will be a parameter of independent interest for other temporal graph problems in the future.

It turns out that our study of temporally Eulerian graphs is closely related to a temporal variant of the Travelling Salesman Problem concerning the exploration of temporal stars via a temporal circuit which starts at the center of the star and which visits all leaves. Akrida, Mertzios and Spirakis introduced this problem and proved it to be $\np$-complete even when the input is restricted to temporal stars in which every edge has at most $k$ appearances for all $k \geq 6$~\cite{starexp}. Although they also showed that the problem is polynomial-time solvable whenever each edge of the input temporal star has at most $3$ appearances, they left open the question of determining the hardness of the problem when each edge has at most $4$ or $5$ appearances. We resolve this open problem in the course of proving our results about temporally Eulerian graphs. Combined with Akrida, Mertzios and Spirakis' results, this gives a complete dichotomy: their temporal star-exploration problem is in $\p$ if each edge has at most $3$ appearances and is $\np$-complete otherwise.

As a potential `island of tractability', Akrida, Mertzios and Spirakis proposed to restrict the input to their temporal star-exploration problem by requiring consecutive appearances of the edges to be evenly spaced (by some globally defined spacing). Using our new notion of interval-membership-width we are able to show that this restriction does indeed yield tractability parameterized by the maximum number of times per edge (thus partially resolving their open problem). Furthermore, we show that a slightly weaker result also holds for the problem of determining whether a temporal graph is temporally Eulerian in the setting with evenly-spaced edge-times. 

Given the success of parameterizations by interval-membership-width when it comes to temporal edge-exploration problems, we then study the related questions of vertex exploration and reachablity. We show that $\minreachdelete$ (a natural vertex-reachability problem) remains hard even on classes of bounded interval-membership-width. This motivates us to introduce a `vertex-variant' of our measure called \emph{vertex-interval-membership-width}. This new measure -- which is lower-bounded by interval-membership-width  -- is more algorithmically powerful: parameterizing by the vertex-variant of our measure puts $\minreachdelete$ in $\fpt$.

\textbf{Outline.} We fix notation and provide background definitions in Section \ref{sec_edge_expl:notation}. We prove our hardness results in Section \ref{sec_edge_expl:hardness}. Section \ref{sec_edge_expl:fpt_algs} contains the definition of interval-membership-width as well as our $\fpt$ algorithms parameterized by this measure. In Section \ref{sec_edge_expl:times_close_together} we show that Akrida, Mertzios and Spirakis' temporal star-exploration problem is in $\fpt$ parameterized by the maximum number of appearances of any edge in the input whenever the input temporal star has evenly-spaced times on all edges. We also show a similar result for our temporally Eulerian problem. In Section~\ref{sec:vertex_version} we introduce a vertex-variant of interval-membership-width and we discuss its algorithmic applications to vertex-exploration and reachability problems. Finally we provide concluding remarks and open problems in Section~\ref{sec_edge_expl:conclusion}.

\section{Background and notation}\label{sec_edge_expl:notation}
For any graph-theoretic notation not defined here, we refer the reader to Diestel's textbook \cite{Diestel2010GraphTheory}; similarly, for any terminology in parameterized complexity, we refer the reader to the textbook by Cygan et al. \cite{cygan2015parameterized}.

The formalism for the notion of dynamic or time-evolving graphs which we adhere to originated from the work of Kempe, Kleinberg, and Kumar \cite{KEMPE2002820}. Formally, if $\tau$ is a function {$\tau: E(G) \to 2^{\mathbb{N}}$} mapping edges of a graph $G = (V(G), E(G))$ to sets of integers, then we call the pair $\mathcal{G} := (G, \tau)$ a \emph{temporal graph}. We shall assume all temporal graphs to be finite and simple in this paper.

For any edge $e$ in $G$, we call the set $\tau(e)$ the \emph{time-set} of $e$. For any time $t \in \tau(e)$ we say that $e$ is \emph{active} at time $t$ and we call the pair $(e,t)$ a \emph{time-edge}. The set of all edges active at any given time $t$ is denoted $E_t(G, \tau) := \{e \in E(G): t \in \tau(e)\}$. The latest time $\Lambda$ for which $E_{\Lambda}(G, \tau)$ is non-empty is called the \emph{lifetime} of a temporal graph $(G, \tau)$ (or equivalently $\Lambda := \max_{e \in E(G)}\max \tau(e)$). Here we will only consider temporal graphs with finite lifetime.

In a temporal graph there are two natural notions of walk: one is the familiar notion of a walk in static graphs and the other is a truly temporal notion where we require consecutive edges in walks to appear at non-decreasing times. Formally, given vertices $x$ and $y$ in a temporal graph $\mathcal{G}$, a \emph{temporal $(x,y)$-walk} is a sequence $W = (e_1,t_1), \ldots, (e_n,t_n)$ of time-edges such that $e_1, \ldots, e_n$ is a walk in $G$ starting at $x$ and ending at $y$ and such that $t_1 \leq t_2 \leq \ldots \leq t_n$. If $n > 1$, we denote by $W - (e_n,t_n)$ the temporal walk $(e_1,t_1), \ldots, (e_{n-1},t_{n-1})$. We call a temporal $(x,y)$-walk \emph{closed} if $x = y$ and we call it a \emph{strict temporal walk} if the times of the walk form a strictly increasing sequence. Hereafter we will assume all temporal walks to be strict.

Recall that an Euler circuit in a static graph $G$ is a circuit $e_1\ldots,e_m$ which traverses every edge of $G$ exactly once. In this section we are interested in the natural temporal analogue of this notion. We point out that, independently and simultaneously to our work here, Marino and Silva also studied temporal variants of the problem of deciding whether a graph is Eulerian~\cite{marino2021k}. 

\begin{definition}
A \emph{temporal Eulerian circuit} in a temporal graph $(G, \tau)$ is a closed temporal walk $(e_1,t_1),\dots,(e_m,t_m)$ such that $e_1\ldots,e_m$ is an Euler circuit in the underlying static graph $G$. If there exists a temporal Eulerian circuit in $(G, \tau)$, then we call $(G, \tau)$ \emph{temporally Eulerian}.
\end{definition}

Note that if $(G, \tau)$ is a temporal graph in which every edge appears at exactly one time, then we can determine whether $(G, \tau)$ is temporally Eulerian in time linear in $\lvert E(G)\rvert$. To see this, note that, since every edge is active at precisely one time, there is only one candidate ordering of the edges (which may or may not give rise to an Eulerian circuit). Thus it is clear that the number of times per edge is relevant to the complexity of the associated decision problem -- which we state as follows. 

\begin{framed}
\noindent
\emph{\textbf{$\tempeuler(k)$}}\\
\textbf{Input:} A temporal graph $(G, \tau)$ where $\lvert  \tau(e) \rvert \leq k$ for every edge $e$ in the graph $G$.\\
\textbf{Question:} Is $(G, \tau)$ \emph{temporally Eulerian}?
\end{framed}

As we mentioned in Section \ref{sec_edge_expl:intro}, here we will show that $\tempeuler(k)$ is related to an analogue of the Travelling Salesman problem on temporal stars \cite{starexp}. This problem (denoted as $\starexp(k)$) was introduced by Akrida, Mertzios and Spirakis \cite{starexp}. It asks whether a given temporal star $(S_n, \tau)$ (where $S_n$ denotes the $n$-leaf star) with at most $k$ times on each edge admits a closed temporal walk starting at the center of the star and which visits every leaf of $S_n$. We call such a walk an \emph{exploration} of $(S_n,\tau)$. A temporal star that admits an exploration is called \emph{explorable}. Formally we have the following decision problem.

\begin{framed}
\noindent
\emph{\textbf{$\starexp(k)$}}\\
\textbf{Input:} A temporal star $(S_n, \tau)$ where $\lvert \tau(e)\rvert \leq k$ for every edge $e$ in the star $S_n$.\\
\textbf{Question:} Is $(S_n, \tau)$ explorable?
\end{framed}

\section{Hardness of temporal edge exploration}\label{sec_edge_expl:hardness}
This section is devoted to showing that $\tempeuler(k)$ is $\np$-complete for all $k$ at least $3$ (Corollary \ref{corollary:temp_Euler_complexity}) and that $\starexp(k)$ is $\np$-complete for all $k$ at least $4$ (Corollary \ref{corollary:NPC-StarExp}). This last result resolves an open problem of Akrida, Mertzios and Spirakis which asked to determine the complexity of $\starexp(4)$ and $\starexp(5)$ \cite{starexp}. 

To show that $\starexp(4)$ is $\np$-hard, we will provide a reduction from the \col problem (see for instance Garey and Johnson \cite{GareyJohnson} for a proof of NP-completeness) which asks whether an input graph $G$ is $3$-colorable.

\begin{framed}
\noindent \emph{\textbf{\col}}\\
\textbf{Input:} A finite simple graph $G$.\\
\textbf{Question:} Does $G$ admit a proper $3$-coloring?
\end{framed}

Throughout, for an edge $e$ of a temporal star $(S_n, \tau)$, we call any pair of times $(t_1,t_2) \in \tau(e)^2$ with $t_1 < t_2$ a \emph{visit of $e$}. We say that $e$ is \emph{visited at $(t_1, t_2)$} in a temporal walk if the walk proceeds from the center of the star along $e$ at time $t_1$ and then back to the center at time $t_2$. We say that two visits $(x_1,x_2)$ and $(y_1,y_2)$ of two edges $e_x$ and $e_y$ are \emph{in conflict} with one another (or that `there is a conflict between them') if there exists some time $t$ with $x_1 \leq t \leq x_2$ and $y_1 \leq t \leq y_2$. Note that a complete set of visits (one visit for each edge of the star) which has no pairwise conflicts is in fact an exploration. 

\begin{theorem}\label{thm:star_exp_hard}
$\starexp(4)$ is $\np$-hard.
\end{theorem}
\begin{proof}
Take any \col instance $G$ with vertices $\{x_1, \ldots, x_n\}$. We will construct a $\starexp(4)$ instance $(S_{p},\tau)$ (where $p = n + 3m$) from $G$. 

\textbf{Defining $\mathbf{S_p}$.} The star $S_p$ is defined as follows: for each vertex $x_i$ in $G$, we make one edge $e_i$ in $S_p$ while, for each edge $x_ix_j$ with $i < j$ in $G$, we make three edges $e_{ij}^{0}$, $e_{ij}^{1}$ and $e_{ij}^{2}$ in $S_p$. 

\textbf{Defining $\mathbf{\tau}$.} For $i \in [n]$ and any non-negative integer $\psi \in \{0,1,2,\dots\}$, let $t^i_{\xi}$ be the integer
\begin{equation}\label{eqn:time_def}
    t^i_{\psi} := 2in^2 + 2\psi(n+1)
\end{equation}
and take any edge $x_jx_k$ in $G$ with $j < k$. Using the times defined in Equation (\ref{eqn:time_def}) and taking $\xi \in \{0,1,2\}$, we then define $\tau(e_i)$ and $\tau(e_{jk}^\xi)$ as
\begin{align}
    \tau(e_i)        &:= \bigl\{t^i_0, t^i_1, t^i_2, t^i_3\bigr\} \label{eqn:vertex_times} \text{ and}\\
    \tau(e^\xi_{jk}) &:= \bigl\{t^j_\xi + 2k -1, 
                          \quad t^j_\xi + 2k, 
                          \quad t^k_\xi + 2j - 1, 
                          \quad t^k_\xi + 2j\bigr\}. \label{eqn:edge_times}
\end{align}
\noindent Note that the elements of these sets are written in increasing order (see Figure \ref{fig:Star_reduction}).

\begin{figure}
    \centering
    \begin{tikzpicture}[-latex , auto , node distance =1 cm and 1 cm , on grid , semithick , state/.style ={ circle, top color =gray, bottom color = white, draw, black, text=black, minimum width =0.2 cm}]
        \node (C)  {};
        \node[state] (X1) [above       = of C] {\tiny{$\mathbf{x_1}$}};
        \node[state] (X2) [below left  = of C] {\tiny{$\mathbf{x_2}$}};
        \node[state] (X3) [below right = of C] {\tiny{$\mathbf{x_3}$}};
        \node[state] (c)   [right = 7cm of C ] {\tiny{$\mathbf{c}$}};
        \node (dummy)      [above = of c     ] {};
        \node[state] (tl)  [left  = 0.5 of dummy ] {};
        \node[state] (tll) [left  = of tl    ] {};
        \node[state] (tr)  [right = 0.5 of dummy ] {};
        \node[state] (trr) [right = of tr    ] {};
        \node (Dummy)      [below = of c] {};
        \node[state] (l1)  [left  = 0.5 of Dummy ] {};
        \node[state] (l2)  [left  = of l1] {};
        \node[state] (l3)  [left  = of l2] {};
        \node[state] (l4)  [left  = of l3] {};
        \node[state] (r1)  [right = 0.5 of Dummy ] {};
        \node[state] (r2)  [right = of r1] {};
        \node[state] (r3)  [right = of r2] {};
        \node[state] (r4)  [right = of r3] {};
        \draw[-] (X1) -- (X2) -- (X3) -- (X1);
        \draw[-] (c) -- (tll);
        \draw[-] (c) -- (tl);
        \draw[-] (c) -- (tr);
        \draw[-] (c) -- (trr);
        \draw[-] (c) -- (r1);
        \draw[-] (c) -- (r2);
        \draw[-] (c) -- (r3);
        \draw[-] (c) -- (r4);
        \draw[-] (c) -- (l1);
        \draw[-] (c) -- (l2);
        \draw[-] (c) -- (l3);
        \draw[-] (c) -- (l4);
        \node [above = 0.4 of tll] {\tiny{$\mathbf{e_1}$}};
        \node [above = 0.4 of tl]  {\tiny{$\mathbf{e^0_{1,2}}$}};
        \node [above = 0.4 of tr]  {\tiny{$\mathbf{e^1_{1,2}}$}};
        \node [above = 0.4 of trr] {\tiny{$\mathbf{e^2_{1,2}}$}};
        \node [below = 0.4 of r1] {\tiny{$\mathbf{e_2}$}};
        \node [below = 0.4 of r2] {\tiny{$\mathbf{e^0_{2,3}}$}};
        \node [below = 0.4 of r3] {\tiny{$\mathbf{e^1_{2,3}}$}};
        \node [below = 0.4 of r4] {\tiny{$\mathbf{e^2_{2,3}}$}};
        \node [below = 0.4 of l1] {\tiny{$\mathbf{e_3}$}};
        \node [below = 0.4 of l2] {\tiny{$\mathbf{e^0_{1,3}}$}};
        \node [below = 0.4 of l3] {\tiny{$\mathbf{e^1_{1,3}}$}};
        \node [below = 0.4 of l4] {\tiny{$\mathbf{e^2_{1,3}}$}};
    \end{tikzpicture}\\[2em]
    \def\variableRightSpacing{1.7cm}
    \def\variableLittleRightSpacingOne{0.5cm}
    \def\variableLittleRightSpacingTwo{0.8cm}
    \begin{tikzpicture}[-latex , auto , node distance =1 cm and 0.5cm , on grid , line width = 0.4pt,  state/.style ={ circle, top color =gray, bottom color = white, draw, black, text=black, minimum width =0.1 cm}]
        \node (x11)                         {\tiny{$\mathbf{t^1_0}$}};
        \node (x12) [right = \variableRightSpacing of x11] {\tiny{$\mathbf{t^1_1}$}};
        \node (x13) [right = \variableRightSpacing of x12] {\tiny{$\mathbf{t^1_2}$}};
        \node (x14) [right = \variableRightSpacing of x13] {\tiny{$\mathbf{t^1_3}$}};
        \node (x21) [right = 0.8 cm of x14] {\tiny{$\mathbf{t^2_0}$}};
        \node (x22) [right = \variableRightSpacing of x21] {\tiny{$\mathbf{t^2_1}$}};
        \node (x23) [right = \variableRightSpacing of x22] {\tiny{$\mathbf{t^2_2}$}};
        \node (x24) [right = \variableRightSpacing of x23] {\tiny{$\mathbf{t^2_3}$}};
        \node (dr1) [below =0.6cm of x11] {};
        \node (r1) [right = \variableLittleRightSpacingOne of dr1] {\tiny{$\mathbf{r_1}$}};
        \node (r2) [right = \variableLittleRightSpacingTwo of r1]    {\tiny{$\mathbf{r_2}$}};
        \node (dr2) [below =0.6cm of x21] {};
        \node (r3) [right = \variableLittleRightSpacingOne of dr2] {\tiny{$\mathbf{r_3}$}};
        \node (r4) [right = \variableLittleRightSpacingTwo of r3]    {\tiny{$\mathbf{r_4}$}};
        \node (dg1) [below =1.2cm of x12] {};
        \node (g1) [right = \variableLittleRightSpacingOne of dg1] {\tiny{$\mathbf{g_1}$}};
        \node (g2) [right = \variableLittleRightSpacingTwo of g1]    {\tiny{$\mathbf{g_2}$}};
        \node (dg2) [below =1.2cm of x22] {};
        \node (g3) [right = \variableLittleRightSpacingOne of dg2] {\tiny{$\mathbf{g_3}$}};
        \node (g4) [right = \variableLittleRightSpacingTwo of g3]    {\tiny{$\mathbf{g_4}$}};
        \node (db1) [below =1.8cm of x13] {};
        \node (b1) [right = \variableLittleRightSpacingOne of db1] {\tiny{$\mathbf{b_1}$}};
        \node (b2) [right = \variableLittleRightSpacingTwo of b1]    {\tiny{$\mathbf{b_2}$}};
        \node (db2) [below =1.8cm of x23] {};
        \node (b3) [right = \variableLittleRightSpacingOne of db2] {\tiny{$\mathbf{b_3}$}};
        \node (b4) [right = \variableLittleRightSpacingTwo of b3]    {\tiny{$\mathbf{b_4}$}};
        \node (l1) [above left = 0.75cm of x11]   {$\tau (e_1)$:};
        \node (l2) [above left = 0.75cm of x21]   {$\tau (e_2)$:};
        \node (lr) [below      = 1.15cm of l1 ]   {$\tau (e^{0}_{12})$:};
        \node (lg) [below      = 0.6cm of lr ]    {$\tau (e^{1}_{12})$:};
        \node (lb) [below      = 0.6cm of lg ]    {$\tau (e^{2}_{12})$:};
        \draw[-, line width = 2pt, red] (x11) -- (x12);
        \draw[-] (x12) -- (x13) -- (x14);
        \draw[-] (x21) -- (x22);
        \draw[-,  line width = 2pt, red] (x22) -- (x23);
        \draw[-] (x23) -- (x24);
        \draw[-] (r1) -- (r2) -- (r3);
        \draw[-,  line width = 2pt, red] (r3) -- (r4);
        \draw[-,  line width = 2pt, red] (g1) -- (g2);
        \draw[-] (g2) -- (g3) -- (g4);
        \draw[-,  line width = 2pt, red] (b1) -- (b2);
        \draw[-] (b2) -- (b3) -- (b4);
    \end{tikzpicture}
    \caption{Top left: $K^3$; we assume the coloring $x_i \mapsto i - 1$. Top right: star constructed from $K^3$. Bottom: times (and corresponding intervals) associated with the edges $e_1$, $e_2$ and $e^0_{1,2}$, $e^1_{1,2}$, $e^2_{1,2}$ (time progresses left-to-right and intervals are not drawn to scale). We write $r_1, r_2, r_3, r_4$ as shorthand for the entries of $\tau(e^0_{1,2})$ (similarly, for $i \in [4]$, we write $g_i$ and $b_i$ with respect to $\tau(e^1_{1,2})$ and $\tau(e^2_{1,2})$). The red and thick intervals correspond to visits defined by the coloring of the $K^3$.}
    \label{fig:Star_reduction}
\end{figure}
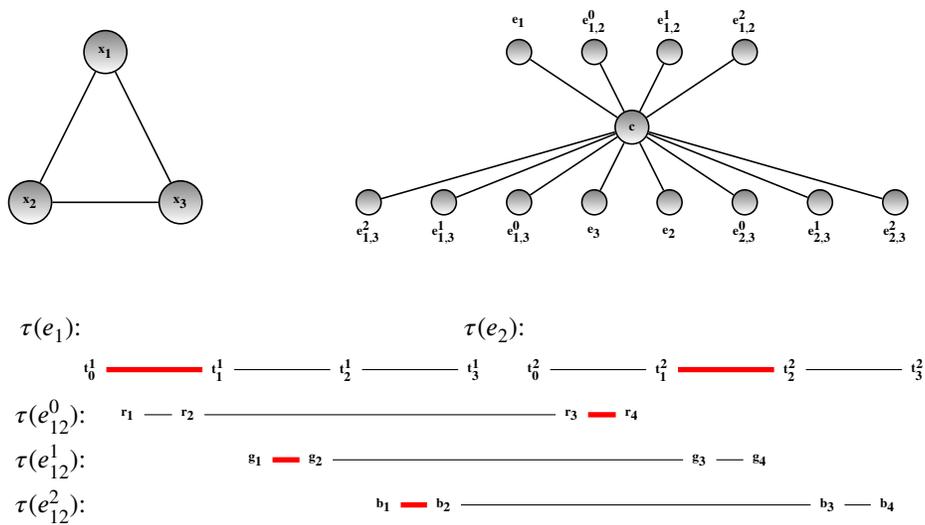
    
Intuitively, the times associated to each edge $e_i \in E(S_p)$ corresponding to a vertex $x_i \in V(G)$ (Equation~(\ref{eqn:vertex_times})) encode the possible colorings of $x_i$ via the three possible starting times of a visit of $e_i$. The three edges $e_{ij}^{0}$, $e_{ij}^{1}$ and $e_{ij}^{2}$ corresponding to some $x_ix_j \in E(G)$ are instead used to `force the colorings to be proper' in $G$. That is to say that, for a color $\xi \in \{0,1,2\}$, the times associated with the edge $e_{ij}^{\xi}$ (Equation~(\ref{eqn:edge_times})) will prohibit us from entering $e_i$ at its $\xi$-th appearance and also entering $e_j$ at its $\xi$-th appearance (i.e. `coloring $x_i$ and $x_j$ the same color').

Observe that the first two times in $\tau(e^\xi_{jk})$ lie within an interval given by consecutive times in $\tau(e_j)$ and that the same holds for the last two times in $\tau(e^\xi_{jk})$ with respect to $\tau(e_k)$ (see Figure \ref{fig:Star_reduction}). More precisely, it is immediate that for $1 \leq j < k \leq n$ and $\xi \in \{0,1,2\}$, we have:
\begin{equation}
    t^j_\xi  \;\:  
    <  \;\:  t^j_\xi + 2k - 1 
    <  \;\:  t^j_\xi + 2k   
    <  \;\:  t^j_{\xi + 1} \label{eqn:hardness_interval}
\end{equation}

\noindent Given this set-up, we will now show that $G$ is a yes instance if and only if $(S_p, \tau)$ is.

\textbf{Suppose $(S_p, \tau)$ is explorable.} Define the coloring (to be shown proper) $c: V(G) \to \{0,1,2\}$ taking each vertex $x_i$ to the color $\xi$ whenever $e_i$ is entered at time $t^i_\xi$ within the exploration of $(S_p, \tau)$ (note that these are the only possible times at which $e_i$ can be entered, since every edge appears at exactly $4$ times). We claim that $c$ is a proper coloring. To see this, suppose on the contrary that there is a monochromatic edge $x_ix_j$ with $i < j$ of color $\xi$ in $G$. Then, this means that $e_i$ was entered at time $t^i_\xi$ and exited at time at least $t^i_{\xi +1}$ and similarly $e_j$ was entered at time $t^j_\xi$ and exited at time at least $t^j_{\xi +1}$. But then, since all times in $\tau(e^\xi_{ij})$ are contained either in the open interval $(t^i_\xi, t^i_{\xi +1}) \subseteq \mathbb{R}$ or the open interval $(t^j_\xi, t^j_{\xi +1}) \subseteq \mathbb{R}$ we know that $e^\xi_{ij}$ cannot be explored (by (\ref{eqn:hardness_interval})). This contradicts the assumption that $(S_p, \tau)$ is explorable, hence $c$ must be a proper coloring. 

\textbf{Conversely, suppose $G$ admits a proper $3$-coloring $c: V(G) \to \{0,1,2\}$.} We define the following exploration of $(S_p, \tau)$ (see Figure \ref{fig:Star_reduction}): 
\begin{itemize}
    \item for every vertex $x_i$ in $G$, if $c(x_i) = \xi$, then visit $e_i$ at $(t^i_\xi, t^i_{\xi + 1})$
    \item for every edge $x_ix_j$ in $G$ with $i < j$ and every color $\xi \in \{0,1,2\}$, define the visit of $e^\xi_{ij}$ as follows: if $c(x_i) \neq \xi$, then visit $e^\xi_{ij}$ at $(t^i_{\xi} + 2j -1, t^i_{\xi} + 2j)$; otherwise visit $e^\xi_{ij}$ at $(t^j_\xi + 2i -1, t^j_\xi + 2i)$. 
\end{itemize}
Our aim now is to show that the visits we have just defined in terms of the coloring $c$ are disjoint (and thus witness the explorability of $(S_p, \tau)$). 

Take any $i < j$ and any $\xi \in \{0,1,2\}$. By our definition of $\tau(e_i)$ and $\tau(e_j)$, we must have $\max \tau(e_i) = t^i_{3} < 2jn^2 = \min \tau(e_j) \text{ whenever } i < j$. Thus we note that there are no conflicts between the visit of $e_i$ and the visit of $e_j$. 
    
Note that, for all $(\xi, \omega) \in \{0,1,2\}^2$ and all pairs of edges $x_ix_j$ and $x_kx_\ell$ in $G$ with $i < j$ and $k < \ell$, the visit $(v_{i,j}, v_{i,j} +1)$ of $e^\xi_{ij}$ is in conflict with the visit  $(v_{k,\ell}, v_{k,\ell} + 1)$ of $e^\omega_{k\ell}$ only if $x_ix_j = x_kx_\ell$.
To see this, observe that the visits of $e^\xi_{ij}$ and $e^\omega_{k\ell}$ both consist of two consecutive times where the first time is odd. Thus we would only have a conflict if $v_{i,j} = v_{k,\ell}$ which can be easily checked to happen only if $i=k$ and $j=\ell$.

Finally we claim that there are no conflicts between the visit of $e^\xi_{ij}$ and the visits of either $e_i$ or $e_j$. To show this, we will only argue for the lack of conflicts between the visits of $e_i$ and $e^\xi_{ij}$ since the same ideas suffice for the $e_j$-case as well. Suppose $c(x_i) = \xi$, then we visit $e^\xi_{ij}$ at $(t^j_\xi + 2i -1, t^j_\xi + 2i)$ and then $t^j_\xi + 2i -1 > t^j_\xi > t^i_3 = \max \tau(e_i)$ since $i < j$ and since $c$ is a proper coloring. Similarly, if $c(x_i) \neq \xi$, then we visit $e^\xi_{ij}$ at $(t^i_{\xi} + 2j -1, t^i_{\xi} + 2j)$. As we observed in Inequality (\ref{eqn:hardness_interval}), we have $t^i_\xi < t^i_{\xi} + 2j -1 < t^i_{\xi} + 2j < t^i_{\xi +1}$. Thus, if $(u_1, u_2)$ is the visit of $e_i$, then either $u_2 < t^i_\xi$ or $t^i_{\xi + 1} < u_1$. In other words, no conflicts arise.

This concludes the proof since we have shown that the visits we assigned to the edges of $S_p$ constitute an exploration of $(S_p, \tau)$.
\end{proof}

Observe that increasing the maximum number of times per edge cannot make the problem easier: we can easily extend the hardness result to any $k' > 4$ by simply adding a new edge with $k'$ times all prior to the times that are already in the star. This, together with the fact that Akrida, Mertzios and Spirakis \cite{starexp} showed that $\starexp(k)$ is in $\np$ for all $k \geq 0$, allows us to conclude the following corollary. 

\begin{corollary}\label{corollary:NPC-StarExp}
For all $k$ at least $4$, $\starexp(k)$ is $\np$-complete.
\end{corollary}

Next we shall reduce $\starexp(k)$ to $\tempeuler(k-1)$. We point out that, for our purposes within this section, only the first point of the statement of the following result is needed. However, later (in the proof of Corollary \ref{corollary:starExp_algorithm}) we shall make use of the properties stated in the second point of Lemma \ref{lemma:Euler_reduction} (this is also why we allow any $k$ times per edge rather than just considering the case $k = 4$). Thus we include full details here.

\begin{lemma}\label{lemma:Euler_reduction}
For all $k \geq 2$ there is a polynomial-time-computable mapping taking every $\starexp(k)$ instance $(S_n, \tau)$ to a $\tempeuler(k-1)$ instance $(D_n, \sigma)$ such that
\begin{enumerate}
    \item $(S_n, \tau)$ is a yes instance for $\starexp(k)$ if and only if $(D_n, \sigma)$ is a yes instance for $\tempeuler(k-1)$ and
    \item $D_n$ is a graph obtained by identifying $n$-copies $\{K^3_1, \ldots, K^3_n\}$ of a cycle on three vertices along one center vertex (see Figure~\ref{fig:Euler_reduction}) and such that 
    \begin{align*}
        \max_{t \in \mathbb{N}} \lvert \{&e \in E(D_n): \min(\sigma(e)) \le t \le \max(\sigma(e))\}\rvert\\
        &\leq 3 \max_{t \in \mathbb{N}} \lvert \{e \in E(S_n): \min(\tau(e)) \le t \le \max(\tau(e))\}\rvert. 
    \end{align*}
\end{enumerate}
\end{lemma}
\begin{proof}
Note that we can assume without loss of generality that: (1) every edge in $S_n$ has exactly $k$-times on each edge and (2) that all times are multiples of $2$. This follows from the fact that we can construct from any $\starexp(k)$-instance $(S_{n'}, u)$ another $\starexp(k)$-instance $(S_n, \tau)$ so that : $(S_n, \tau)$ is explorable if and only if $(S_{n'}, u)$ also is and every time in $(S_n, \tau)$ satisfies conditions (1) and (2).
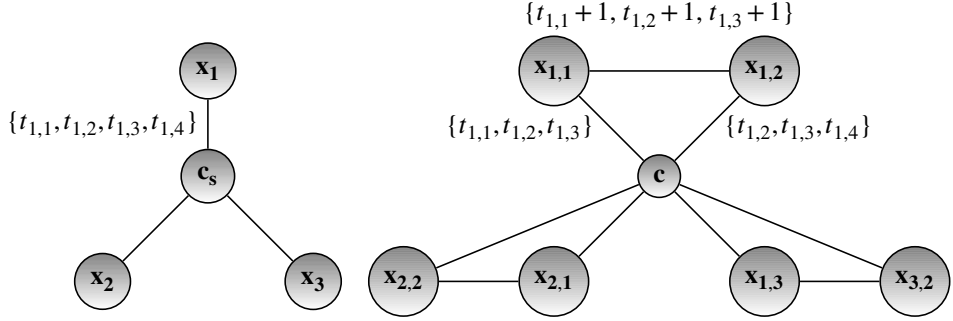
\begin{figure}
    \centering
    \begin{tikzpicture}[-latex , auto , node distance =1.4 cm and 1.4 cm , on grid , semithick , state/.style ={ circle, top color =gray, bottom color = white, draw, black, text=black, minimum width =0.2 cm}]
        \node[state] (C)  {$\mathbf{c_s}$};
        \node[state] (X1) [above =of C]{$\mathbf{x_1}$};
        \node[state] (X2) [below left =of C] {$\mathbf{x_2}$};
        \node[state] (X3) [below right =of C] {$\mathbf{x_3}$};
        \draw[-] (C) -- (X1) node [midway] {$\{t_{1,1}, t_{1,2}, t_{1,3}, t_{1,4}\}$};
        \draw[-] (C) -- (X2);
        \draw[-] (C) -- (X3);
        \node[state] (CC) [right =6cm of C] {$\mathbf{c}$};
        \node[state] (Y11) [above left =of CC]{$\mathbf{x_{1,1}}$};
        \node[state] (Y12) [above right =of CC]{$\mathbf{x_{1,2}}$};
        \node[state] (Y21) [below left =of CC] {$\mathbf{x_{2,1}}$};
        \node[state] (Y22) [left =2cm of Y21]{$\mathbf{x_{2,2}}$};
        \node[state] (Y31) [below right =of CC] {$\mathbf{x_{1,3}}$};
        \node[state] (Y32) [right =2cm of Y31]{$\mathbf{x_{3,2}}$};
        \draw[-] (CC) -- (Y11) node [midway, left = 3pt   ] {$\{t_{1,1}, t_{1,2}, t_{1,3}\}$};
        \draw[-] (CC) -- (Y12) node [midway, right = 3pt  ] {$\{t_{1,2}, t_{1,3}, t_{1,4}\}$};
        \draw[-] (Y11) -- (Y12) node [midway, above = 12pt] {$\{t_{1,1} + 1, \: t_{1,2} + 1, \: t_{1,3} + 1\}$};
        \draw[-] (CC) -- (Y21);
        \draw[-] (CC) -- (Y22);
        \draw[-] (Y21) -- (Y22);
        \draw[-] (CC) -- (Y31);
        \draw[-] (CC) -- (Y32);
        \draw[-] (Y31) -- (Y32);
        
    \end{tikzpicture}
    \caption{Building $(D_3, \sigma)$ from $(S_3, \tau)$. The times along edges are drawn only for the edge $c_sx_1$ in $S_3$ and for its corresponding $3$-cycle $cx_{1,1}x_{1,2}$ in $D_3$. Since $t_{1,1}$, $t_{1,2}$, $t_{1,3}$ and $t_{1,4}$ are all multiples of $2$, we know that $t_{1,j} < t_{1,j +1} < t_{1,j+1}$ for all $j \in [3]$. Thus the reduction associates the visit $(t_s, t_e)$ of $c_sx_1$ in the star to exploration $(t_s, t_s + 1, t_e)$ of the $3$-cycle corresponding to $c_sx_1$ in $D_3$.}
    \label{fig:Euler_reduction}
\end{figure}

Now we will show how to construct a {$\tempeuler(k-1)$-instance} $(D_n, \sigma)$ from $(S_n, \tau)$ such that $(D_n, \sigma)$ is temporally Eulerian if and only if $(S_n, \tau)$ is explorable (see Figure \ref{fig:Euler_reduction}). Throughout, denote the vertices of the $i$-th $3$-cycle $C^3_i$ of $D_n$ by $\{c,  x_{i,1}, x_{i,2}\}$ and let its edges be $f_{i,1} = cx_{i,1}$, $f_{i,2} = x_{i,1}x_{i,2}$ and $f_{i,3} = x_{i,2}c$. For every $i \in [n]$ with $\tau(e_i) = \{t_1, \ldots, t_k\}$ where $t_1 < t_2 < \dots < t_k$, define the map $\sigma : E(D_n) \to 2^{\mathbb{N}}$ as: 
    \begin{align*}
        \sigma(f_{i,1}) &:= \{t_1, \ldots, t_{k-1}\}, \\
        \sigma(f_{i,2}) &:= \{t_1 + 1, \ldots,  t_{k-1} + 1\}, \\
        \sigma(f_{i,3}) &:= \{t_2, \ldots, t_{k}\}. 
    \end{align*}
Note that $\lvert \sigma(f_{i,1})\rvert = \lvert \sigma(f_{i,2})\rvert = \lvert \sigma(f_{i,3})\rvert = k-1$. Now suppose $(S_n, \tau)$ is a yes-instance witnessed by the sequence $\mathcal{V}$ of visits $\mathcal{V} := (x_1,y_1), \ldots, (x_n, y_n)$ of the edges $e_1, \ldots, e_n$ of $S_n$ and observe that $y_i < x_{i+1}$ for all $i \in [n-1]$. We claim that the sequence of time-edges \[(f_{1,1}, x_1), (f_{1,2}, x_1 + 1), (f_{1,3}, y_1), \ldots, (f_{n,1}, x_n), (f_{n,2}, x_n + 1), (f_{n,3}, y_n)\] is a temporal Eulerian circuit in $G$. To see this, recall that $y_j < x_{j+1}$ (for $j \in [n-1]$) and note that: 
\begin{enumerate}
    \item by definition $f_{1,1}, f_{1,2}, f_{1,3}, \ldots, f_{n,1}, f_{n,2}, f_{n,3}$ is an Eulerian circuit in the underlying static graph $D_n$ (i.e. we walk along each $3$-cycle in turn) and 
    \item $x_i < x_i + 1 < y_i$ for all $i \in [n]$ since we assumed that $x_i, y_i \in 2\mathbb{N}$.
\end{enumerate}
Conversely, suppose $(D_n, \sigma)$ is a yes-instance and let this fact be witnessed by the temporal Eulerian circuit $\mathcal{K}$. Recall that a temporal Eulerian circuit induces an Eulerian circuit in the underlying static graph. Thus, since every Eulerian circuit in $D_n$ must run through each $3$-cycle, we know that $\mathcal{K}$ must consist -- up to relabelling of the edges -- of a sequence of time-edges of the form 
\begin{align*}
    \mathcal{K} := (f_{1,1}, x_{1,1}), (f_{1,2}, x_{1,2}), (f_{1,3}, y_{1,3})&,\;\; (f_{2,1}, x_{2,1}), (f_{2,2}, x_{2,2}), (f_{2,3}, y_{2,3}), \\
    \ldots&,\;\; (f_{n,1}, x_{n,1}), (f_{n,2}, x_{n,2}), (f_{n,3}, y_{n,3}).
\end{align*}
It follows immediately from the definition of $(D_n, \sigma)$ that visiting each edge $e_j$ in $S_n$ at $(x_{j,1}, y_{j,3})$ constitutes an exploration of $(S_n, \tau)$, as desired.
\end{proof} 

Since $\tempeuler(k)$ is clearly in $\np$ (where the circuit acts as a certificate), our desired $\np$-completeness result follows immediately from Lemma~\ref{lemma:Euler_reduction} and Corollary~\ref{corollary:NPC-StarExp}. 
\begin{corollary}\label{corollary:temp_Euler_complexity}
$\tempeuler(k)$ is $\np$-complete for all $k$ at least $3$.
\end{corollary}

As we noted earlier, $\tempeuler(1)$ is trivially solvable in time linear in the number of edges of the underlying static graph. Although our proof leaves open the $k=2$ case, Marino and Silva closed this gap showing that $\tempeuler(k)$ is $\np$-complete for all $k \geq 2$ (thus resolving an open problem from a previous version of this paper).

Observe that the reduction in Lemma \ref{lemma:Euler_reduction} rules out $\fpt$ algorithms with respect to many standard parameters describing the structure of the underlying graph (for instance the path-width is $2$ and feedback vertex number\footnote{Recall that a \emph{feedback vertex set} in a  graph (resp. directed graph) $G$ is a vertex subset $S \subseteq V(G)$ such that $G - S$ is a forest (resp. directed acyclic graph). The \emph{feedback vertex number} of a graph $G$ is the minimum number of vertices needed to form a feedback vertex set in $G$.} is $1$). In fact we can strengthen these intractability results even further by showing that $\tempeuler(k)$ is hard even for instances whose underlying static graph has vertex-cover number\footnote{Recall that a \emph{vertex-cover} in a graph $G$ is a vertex-subset $S \subseteq V(G)$ such that every edge in $G$ is incident with at least one vertex in $S$. The \emph{vertex-cover number} of a graph $G$ is the minimum number of vertices needed to form a vertex-cover of $G$.} $2$. This motivates our search in Section \ref{sec_edge_expl:fpt_algs} for parameters that describe the structure of the times assigned to edges rather than just the underlying static structure. 

Notice that this time we will reduce from $\starexp(k)$ to $\tempeuler(k)$ (rather than from $\starexp(k+1)$ as in Lemma \ref{lemma:Euler_reduction}), so, in contrast to our previous reduction (Lemma \ref{lemma:Euler_reduction}), the proof of the following result cannot be used to show hardness of $\tempeuler(3)$. 
\begin{theorem}\label{thm:reduction_VC}
For all $k \geq 4$, the $\tempeuler(k)$ problem is $\np$-complete even on temporal graphs whose underlying static graph has vertex-cover number $2$.
\end{theorem}
\begin{proof}
Take any $\starexp(k)$ instance $(S_n,\tau)$ and assume that $n$ is even (if not, then simply add a dummy edge with all appearances strictly after the lifetime of the graph). Denoting by $c$ the center of $S_n$ and by $x_1, \ldots, x_n$ its leaves, let $S^{c_1,c_2}_n$ be the double star constructed from $S_n$ by splitting $c$ into two twin centers; stating this formally, we define $S^{c_1,c_2}_n$ as
\[S^{c_1,c_2}_n = (\{c_1,c_2,x_1,\ldots, x_n\}, \{c_ix_j: i \in [2] \text{ and } j \in [n]\}).\]
Notice that, since $n$ is even, $S^{c_1,c_2}_n$ is Eulerian and notice that the set $\{c_1,c_2\}$ is a vertex cover of $S^{c_1,c_2}_n$. 

Defining $\sigma: E(S^{c_1,c_2}_n) \to 2^{\mathbb{N}}$ for all $i \in [2]$ and $j \in [n]$ as $\sigma(c_ix_j) = \tau(cx_j)$, we claim that the temporal graph $(S^{c_1,c_2}_n, \sigma)$ is temporally Eulerian if and only if $(S_n, \tau)$ is explorable. 

Suppose that $(S_n, \tau)$ is explorable and let this be witnessed by the sequence of visits $(s_1,t_1), \ldots, (s_n, t_n)$. Then it follows immediately by the definition of $\sigma$ that the following sequence of time-edges is a temporal circuit in $(S^{c_1,c_2}_n, \sigma)$:
\[
    (c_1x_1, s_1),(x_1c_2, t_1),(c_2x_2, s_2)(x_2c_1, t_2), (c_1x_3, s_3),(x_3c_2, t_3), \ldots,  (c_2x_n, s_n)(x_nc_1, t_n).
\]
To see this, note that this clearly induces an Eulerian circuit in the underlying static graph $S^{c_1,c_2}_n$; furthermore, since $(s_1,t_1), \ldots, (s_n, t_n)$ is an exploration in $(S_n, \tau)$, it follows that $s_1 < t_1 < s_2 < t_2 < \ldots < s_n < t_n$, as desired. 

Suppose now that $(S^{c_1,c_2}_n, \sigma)$ is temporally Eulerian and that this fact is witnessed (without loss of generality -- up to relabeling of vertices) by the temporal Eulerian circuit 
\[
    (c_1x_1, s_1),(x_1c_2, t_1),(c_2x_2, s_2)(x_2c_1, t_2), (c_1x_3, s_3),(x_3c_2, t_3), \ldots,  (c_2x_n, s_n)(x_nc_1, t_n).
\]
Then, by the definition of $\sigma$ in terms of $\tau$ and by similar arguments to the previous case, we have that $(s_1,t_1), \ldots, (s_n, t_n)$ is an exploration of $(S_n,\tau)$.
\end{proof}

\section{Interval-membership-width}\label{sec_edge_expl:fpt_algs}

As we saw in the previous section, both $\tempeuler(k)$ and $\starexp(k+1)$ are $\np$-complete for all $k \geq 3$ even on instances whose underlying static graphs are very sparse (for instance even on graphs with vertex cover number $2$). Clearly this means that any useful parameterization must take into account the \emph{temporal structure} of the input. As we discussed previously, other authors have already proposed measures of this kind such as the temporal feedback vertex number \cite{temporalFeedbackVertexSet} or temporal analogues of tree-width \cite{temp_tw}. However these measures are all bounded on temporal graphs for which the underlying static graph has bounded feedback vertex number and tree-width respectively. Our reductions therefore show that $\tempeuler(k)$ is para-$\np$-complete with respect to these parameters. Thus we do indeed need some new measure of temporal structure. To that end, here we introduce such a parameter called \emph{interval-membership-width} which depends only on temporal structure and not on the structure of the underlying static graph. Parameterizing by this measure, we will show that both $\tempeuler(k)$ and $\starexp(k)$ lie in $\fpt$. 

To first convey the intuition behind our width measure, consider again the simplest case of the problem of deciding whether a temporal graph is temporally Eulerian; namely: consider $\tempeuler(1)$. As we noted earlier, this problem is trivially solvable in time linear in $\lvert E(G)\rvert$. The same is true for any $\tempeuler(k)$-instance $(G, \tau)$ in which every edge is assigned a `private' interval of times: that is to say that, for all distinct edges $e$ and $f$ in $G$, either $\max \tau (f) < \min \tau (e)$ or $\max \tau (e) < \min \tau (f)$. This holds because, on instances of this kind, there is only one possible relative ordering of edges available for an edge-exploration. It is thus natural to expect that, for graphs whose edges have intervals that are `almost private' (defined formally below), we should be able to deduce similar tractability results. 

Towards a formalization of this intuition, suppose that we are given a temporal graph $(G, \tau)$ which has precisely two edges $e$ and $f$ such that there is a time $t$ with $\min \tau (e) \leq t \leq \max \tau (e)$ and $\min \tau (f) \leq t \leq \max \tau (f)$.
It is easy to see that the $\tempeuler(k)$ problem is still tractable on graphs such as $(G, \tau)$ since there are only two possible relative edge-orderings for an edge exploration of $(G, \tau)$ (depending on whether we choose to explore $e$ before $f$ or $f$ before $e$). These observations lead to the following definition of \emph{interval-membership-width} of a temporal graph (see Figure~\ref{fig:imw_example}). 
\begin{definition}\label{definition:width_measure}
The \emph{interval-membership  sequence} of a temporal graph $(G, \tau)$ is the sequence  $(F_t)_{t \in [\Lambda]}$ of edge-subsets of $G$ where $F_t:= \{e \in E(G): \min \tau(e) \leq t \leq \max \tau(e)\}$ and $\Lambda$ is the lifetime of $(G, \tau)$. The \emph{interval-membership-width} of $(G, \tau)$ is the integer $\imw(G, \tau) := \max_{t \in \mathbb{N}} \lvert F_t\rvert$.
\end{definition}
\begin{figure}
    \centering
    \begin{tikzpicture}[-latex , auto , node distance =1.4 cm and 1.4 cm , on grid , semithick , state/.style ={ circle, top color =gray, bottom color = white, draw, black, text=black, minimum width =0.2 cm}]
        \node[state] (c)   [] {$c$};
        \node (dummy)      [above = of c     ] {};
        \node (tl)  [left  = 0.5 of dummy ] {};
        \node[state] (tll) [left  = of tl    ] {$w$};
        \node (tr)  [right = 0.5 of dummy ] {};
        \node[state] (trr) [right = of tr    ] {$x$};
        \node (Dummy)      [below = of c] {};
        \node (l1)  [left  = 0.5 of Dummy ] {};
        \node[state] (l2)  [left  = of l1] {$z$};
        \node (r1)  [right = 0.5 of Dummy ] {};
        \node[state] (r2)  [right = of r1] {$y$};
        \draw[-] (c) -- node[left =0.15 cm] {$1,9$} (tll);
        \draw[-] (c) -- node[left =0.15 cm] {$3,5$} (trr);
        \draw[-] (c) -- node[left =0.15 cm] {$4,6$} (r2);
        \draw[-] (c) -- node[left =0.15 cm] {$7,8,9$} (l2);
        
        \node (my_text)  [below  = 2cm of c ] {$(S_4, \tau)$};
    \end{tikzpicture}
    \caption{A temporal star $(S_4, \tau)$ with interval-membership sequence: $F_1 = F_2 = \{cw\}$, $F_3 = \{cw, cx\}$, $F_4 = F_5 = \{cw, cx, cy\}$, $F_6 = \{cw, cy\}$ and $F_7= F_8 = F_9 = \{cw, cz\}$.}
    \label{fig:imw_example}
\end{figure}
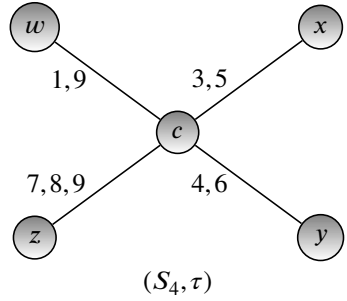
Note that a temporal graph has unit interval-membership-width if and only if every edge is active at times spanning a `private interval'. Furthermore, we point out that the interval-membership  sequence of a temporal graph is not the same as the sequence $(E_t(G, \tau))_{t \in \mathbb{N}}$. In fact, although $\max_{t \in \mathbb{N}}\lvert E_t(G, \tau)\rvert \leq \imw(G, \tau)$, there exist classes $\mathcal{C}$ of temporal graphs with unbounded interval-membership-width but such that every temporal graph in $\mathcal{C}$ satisfies the property that at most one edge is active at any given time. To see this consider any graph $H$ with edges $e_1, \ldots, e_m$ and let $(H, \nu)$ be the temporal graph defined by $\nu (e_i) := \{i,m+i\}$. Clearly $\max_{i \in \mathbb{N}} \lvert E_i(H, \nu)\rvert = 1$, but we have $\imw(H, \nu) = m$. 

Before continuing, we first note how to compute the interval-membership  sequence of a temporal graph $(G, \tau)$ with lifetime $\Lambda$ in $\mathcal{O}\bigl(\imw(G, \tau) \Lambda) \bigr)$ time. We point out now that, since this running time is linear in $\Lambda$, all of the running times in our following (in this section and the next) algorithmic results parameterized by interval-membership-width also include the time needed to compute the interval-membership sequence. 

\begin{lemma}\label{lemma:construction_algorithm}
There is an algorithm that, given a temporal graph $(G, \tau)$ with every edge active at-least once and with lifetime $\Lambda$, computes the interval-membership sequence of $(G, \tau)$ in time $\mathcal{O}(w \Lambda)$ where $w = \imw(G, \tau)$.
\end{lemma}
\begin{proof}
Supposing $E(G) = \{e_1, \dots e_m\}$, consider the following algorithm:
\begin{itemize}
    \item initialize a list $(F_t)_{t \in [\Lambda]}$ with $F_t = \emptyset$ for all $t$,
    \item for each edge $e \in E(G)$, 
    \begin{enumerate}[label=\textbf{S\arabic*}]
   	 \item compute $m_e = \min \tau(e)$ and $M_e = \max \tau(e)$  \label{step:mk_bounds}
	 \item add $e$ to every set $F_i$ with $m_e \leq i \leq M_e$. \label{step:fill}
    \end{enumerate}
\end{itemize}
For each edge $e$, Step~\ref{step:mk_bounds} takes $\mathcal{O}(\lvert \tau(e)\rvert)$ time while Step~\ref{step:fill} takes $\mathcal{O}(\lvert \{t : e \in F_t\}\rvert)$ time. Thus, since $\tau(e) \subseteq \{t : e \in F_t\}$, steps~\ref{step:mk_bounds} and~\ref{step:fill} together take time $\mathcal{O}(\lvert \{t : e \in F_t\}\rvert)$. Summing over all edges of $G$, we have that, since every edge is active at-least once, the whole algorithm runs in time
\[\mathcal{O}\Bigl( \sum_{e \in E(G)}\lvert \{t : e \in F_t\}\rvert\Bigr) = \mathcal{O}\Bigl( \sum_{t \in [\Lambda]}\lvert F_t\rvert\Bigr) = \mathcal{O}(w \Lambda).\]
\end{proof}

Armed with the notion of interval-membership-width, we will now show that both $\tempeuler(k)$ and $\starexp(k)$ are in $\fpt$ when parameterized by this measure. We will do so first for $\tempeuler(k)$ (Theorem \ref{thm:temp_Euler_dyn_prog}) and then we will leverage the reduction of Lemma  \ref{lemma:Euler_reduction} to deduce the fixed-parameter-tractability of $\starexp(k)$ as well (Corollary \ref{corollary:starExp_algorithm}). 

\begin{theorem}\label{thm:temp_Euler_dyn_prog}
There is an algorithm that, given any temporal graph $(G, \tau)$ with lifetime $\Lambda$, decides whether $(G, \tau)$ is a yes-instance of $\tempeuler(k)$ in time $\mathcal{O}(w^3 2^w \Lambda)$ where $w = \imw(G, \tau)$ is the interval-membership-width of $(G, \tau)$.
\end{theorem}
\begin{proof}
Let  $(F_t)_{t \in [\Lambda]}$ be the interval-membership sequence of $(G, \tau)$ and suppose without loss of generality that $F_1$ is not empty.

We will now describe an algorithm that proceeds by dynamic programming over the sequence $(F_i)_{i \in [\Lambda]}$ to determine whether $(G, \tau)$ is temporally Eulerian. For each set $F_i$ we will compute a set $L_i \subseteq F_i^{\{0,1\}} \times V(G) \times V(G)$ consisting of triples of the form $(f, s, x)$ where $s$ and $x$ are vertices in $G$ and $f$ is a function mapping each edge in $F_i$ to an element of $\{0,1\}$. Intuitively each entry $(f,s,x)$ of $L_i$ corresponds to the existence of a temporal walk starting at $s$ and ending at $x$ at time at most $i$ and such that, for any edge $e\in F_i$, we will have $f(e) = 1$ if and only if $e$ was traversed during this walk.  

We will now define the entries $L_i$ recursively starting from the dummy set $L_0 := \{(\mathbf{0}, x, x) :\exists e \in F_1 \text{ incident with } x\}$ where $\mathbf{0}: e \in F_1 \mapsto 0$ is the function mapping every element in $F_1$ to $0$. Take any $(f, s, y)$ in $F_i^{\{0,1\}} \times V(G) \times V(G)$. For $(f, s, y)$ to be in $L_i$ we will require there to be an entry $(g, s, x)$ of $L_{i-1}$ such that
\begin{equation}\label{equation:TE1}
    g(e) = 1 \text{ for all } e \in F_{i-1} \setminus F_i 
\end{equation}
and such that the one of the following cases holds: either
\begin{enumerate}[label=\textbf{C\arabic*}]
    \item $y = x$ and $f(e) = 1$ if and only if $e \in F_{i-1} \cap F_i$ and $g(e) = 1$, \label{case:stay_put} \\
    or
    \item there exists an edge $xy$ in $G$ such that: \label{case:move}
        \begin{enumerate}[label=\textbf{C2.P\arabic*}]
            \item $xy \in E_i(G, \tau) \setminus \{e \in F_i: g(e) = 1\}$ and \label{property:C2.1}
            \item $f(e) = 1$ if and only if $g(e) = 1$ or $e = xy$. \label{property:C2.2}
        \end{enumerate}
\end{enumerate}

The Cases \ref{case:stay_put} and \ref{case:move} correspond to the the two available choices we have when extending a temporal $(s,x)$-walk at time $i$: either we stay put at $x$ (Case \ref{case:stay_put}) or we find some new edge $xy$ active at time $i$ (Case \ref{case:move}) which has never been used before (Property \ref{property:C2.1}) and add it to the walk (Property \ref{property:C2.2}). Equation (\ref{equation:TE1}) ensures that we filter out partial solutions that we already know cannot be extended to a Eulerian circuit. To see this, note that, if an edge $e$ will never appear again after time $i-1$ and we have $g(e) = 0$, then there is no way of extending the temporal walk represented by the triple $(g, s, x)$ to an Eulerian circuit in $(G, \tau)$ because one edge will always be left out (namely the edge~$e$). 

We claim that the input $(G,\tau)$ is temporally Eulerian if and only if $L_\Lambda$ contains an entry $(\mathbf{1}, s, x)$ with $s=x$ and such that $\mathbf{1}$ is the constant all-$1$ function $\mathbf{1}: y \in F_\Lambda \mapsto 1$. To show this, we will prove the following stronger claim. 
\begin{claim}\label{claim:correctness_alg}
    For all $i \in [\Lambda]$, $L_i$ contains an entry $(f, s, x)$ if and only if there exists a temporal walk $(e_1,t_1) \ldots (e_{p},t_{p})$ starting at $s$ and ending at $x$ with $t_p \leq i$ and in which no edge is repeated and such that:
    \begin{enumerate}[label=\textbf{IH\arabic*}]
        \item $(F_1 \cup \dots \cup F_{i-1}) \setminus F_i \subseteq \{e_1, \ldots, e_{p-1}\}$ (i.e. every edge whose last appearance is before time $i$ is traversed by the walk) and \label{prop:IH1}
        \item for all $e \in F_i$, we have $f(e) = 1$ if $e \in \{e_1, \ldots, e_{p}\}$ and $f(e) = 0$ otherwise (i.e. $f$ correctly records which edges in $F_i$ have been used in a walk). \label{prop:IH2} 
    \end{enumerate}
\end{claim}
\begin{proof}[Proof of Claim~\ref{claim:correctness_alg}]
We show this by induction on $i$. The Claim holds trivially for $i = 0$, so suppose now that we are at some time $i > 0$ and hypothesise that the Claim holds for time $i-1$. Furthermore denote by $W_i(f,s,x)$ the set of all temporal temporal $(s,x)$-walks $(e_1,t_1) \ldots (e_{p},t_{p})$ with $t_p \leq i$ which satisfy Properties \ref{prop:IH1} and \ref{prop:IH2}.

$\mathbf{(\Longrightarrow)}$ First we will show that if $(f, s, y)$ is in $L_i$, then $W_i(f,s,x)$ is non-empty. By the construction of $L_i$, we know that, for $(f, s, y)$ to be in $L_i$, there must have been an element $(g, s, x)$ of $L_{i-1}$ satisfying Equation (\ref{equation:TE1}) from which we built $(f,s,y)$ according to either Case \ref{case:stay_put} or Case \ref{case:move}. 

Suppose we applied Case \ref{case:stay_put} to add $(f,s,y)$ to $L_i$ (i.e. we `extended' some walk in $W_{i-1}(g,s,x)$ by deciding not to move). Then $x = y$ and we know that $f(e) = 1$ if and only if $g(e) = 1$. Notice that any walk corresponding to $(f,s,y)$ cannot fail to visit some edge in $E_{i-1}(G, \tau)$ that will never again be active after time $i-1$ since we know that $(g, s, x)$ satisfies Equation (\ref{equation:TE1}). In particular $(f, s, y)$ satisfies Property \ref{prop:IH1} (since $g$ does). Furthermore, since $f(e) = 1$ if and only if $g(e) = 1$ and since $g$ satisfies Property \ref{prop:IH2} (by induction), we know that $f$ must also satisfy Property \ref{prop:IH2}. Thus we have shown that, if we applied Case \ref{case:stay_put} to add $(f,s,y)$ to $L_i$, then $W_i(f,s,y) \neq \emptyset$.

Suppose instead that we applied Case \ref{case:move} to add $(f,s,y)$ to $L_i$. In other words suppose we found an edge $xy$ active at time $i$ with which we wish to extend some walk $W := (e_1,t_1), \ldots, (e_{p},t_{p})$ in $W_{i-1}(g,s,x)$ which starts at $s$ and ends at $x$. Note that we can infer that $W' := (e_1,t_1) \ldots (e_{p-1},t_{p})(xy,i)$ is a valid temporal $(s,y)$-walk with no repeated edges since: 
\begin{itemize}
    \item $W$ has no repeated edges (by the induction hypothesis) and 
    \item $xy$ was not traversed by $W$ (by Property \ref{property:C2.1}) and 
    \item $t_p \leq i -1$ (since $W \in W_{i-1}(g,s,x)$).
\end{itemize} Thus the fact that $g$ satisfies equation (\ref{equation:TE1}) combined with the induction hypothesis implies that every edge whose last appearance is before time $i$ is traversed by $W'$ (i.e. $W'$ satisfies Property \ref{prop:IH1}). Furthermore $f$ satisfies Property \ref{prop:IH2} since $g$ does and since $f(e) = 1$ if and only if $g(e) = 1$ or $e = xy$ (by Property \ref{property:C2.2}). Thus we have shown that, if $(f, s, y) \in L_i$, then $W_i(f,s,y) \neq \emptyset$.

$\mathbf{(\Longleftarrow)}$ Conversely, we will now show that, if $W_i(f,s,y)$ is non-empty, then $(f,s,y) \in L_i$. Let $W'$ be an element of $W_i(f,s,y)$ and let $(xy, j)$ be the last time-edge traversed by $W'$ (note $j \leq i$). 

If $j < i$ then, by the induction hypothesis, there exists an entry $(g,s,y) \in L_{i-1}$ with $ W' \in W_{i-1}(g,s,y)$. But then by the construction of $L_i$ from $L_{i-1}$ we have that $(f,s,y) \in L_i$. 

Thus suppose $j = i$. Then $W' - (xy, j)$ is a temporal $(s,x)$-walk ending at time at most $i-1$ satisfying Property \ref{prop:IH1}. Furthermore, by the induction hypothesis, there must be a $(g,s,x) \in L_{i-1}$ which satisfies Equation (\ref{equation:TE1}) and such that $W' - (xy, j) \in W_{i-1}(g,s,x)$. Now note that, since $(f,s,y)$ satisfies Properties \ref{prop:IH1} and \ref{prop:IH2}, we have that Properties \ref{property:C2.1} and \ref{property:C2.2} hold as well: thus $(f,s,y) \in L_i$.
\end{proof}

Finally we consider the running time of the algorithm. First of all notice that we can compute $L_{i+1}$ from $L_i$ in time at most $(\lvert E_{i+1}(G, \tau)\rvert + 1) \cdot \lvert L_i\rvert \leq (\lvert F_{i+1}\rvert + 1) \cdot \lvert L_i\rvert$. To see this, note that we construct the elements of $L_{i+1}$ by iterating through the elements of $L_{i}$ and considering for each one the $\lvert F_{i+1}\rvert + 1$ ways of taking a next step in a temporal walk at time $i$. Since we perform this computation $\Lambda$ times and since $\lvert F_i\rvert \leq \imw(G, \tau) = w$, the whole algorithm runs in time $\mathcal{O}(w \Lambda\max_{i \in \Lambda}\lvert L_i\rvert)$. Thus all that remains to be shown to complete the proof is that $\lvert L_i\rvert$ is $\mathcal{O}(w^2 2^w)$ for all $i$. Note that, from its definition, we already know that $L_i$ has cardinality at most $\mathcal{O}(2^wn^2)$ since $L_i \subseteq F_i^{\{0,1\}} \times V(G) \times V(G)$. To improve this bound, we will show that the following two statements hold: 
\begin{enumerate}[label=\textbf{RT\arabic*}]
    \item there exists a time $t$ such that every temporal Eulerian circuit in $(G, \tau)$ must start with a vertex incident with an edge in the bag $F_t$ of the interval-membership sequence of $(G, \tau)$; \label{running_bound_claim_start_vertex}
    \item for all $i$, let $\mathcal{X}_i \subseteq V(G)$ be the set of vertices of $G$ defined as $\mathcal{X}_i := \{x \in V(G): (f,s,x) \in L_i\}$; then $\mathcal{X}_i$ has cardinality at most $4w$, where $w = \imw(G, \tau)$. \label{running_bound_claim_end_vertex}
\end{enumerate}
To see why it suffices to prove claims \ref{running_bound_claim_start_vertex} and \ref{running_bound_claim_end_vertex}, notice that they imply that we not only have $L_i \subseteq F_i^{\{0,1\}} \times V(G) \times V(G)$ (which was how we defined $L_i$ in the first place) but in fact there must always exist a $t \in [\Lambda]$ and a subset $\mathcal{X}_i \subseteq V(G)$ (for all $i$) such that $L_i$ is always of the form \[L_i \subseteq F_i^{\{0,1\}} \times V(F_t) \times \mathcal{X}_i\] where both $\lvert V(F_t)\rvert$ and $\lvert \mathcal{X}_i\rvert$ are $\mathcal{O}(w)$. This would clearly then imply that $\lvert L_i\rvert$ is $\mathcal{O}(w^2 2^w)$ for all $i$, as desired.

\begin{proof}[Proof of Claim~\ref{running_bound_claim_start_vertex}]
Choose $t \in \mathbb{N}$ be greatest possible such that $\bigcup_{j \in [t]} F_j \subseteq F_t$. Suppose by way of contradiction that there exists a temporal Eulerian circuit that starts at a vertex $s$ with $s$ \emph{not} incident with any edge in $F_t$. Let $t'$ be the earliest time such that the bag $F_{t'}$ contains an edge which which $s$ is incident. 

Notice that, since $t$ was chosen greatest possible such that $\bigcup_{j \in [t]} F_j \subseteq F_t$ and since $s$ is not incident with any edge in $F_t$, it follows that $t' > t$ and that there exists an edge $e \in F_t \setminus F_{t'}$. But then we have a contradiction since $\max(\tau(e)) \leq t < t'$ and, by time $t'$, $e$ has not yet been visited by the temporal Eulerian circuit starting at $s$ (i.e. any such circuit never vists the edge $e$). 
\end{proof}

\begin{proof}[Proof of Claim~\ref{running_bound_claim_end_vertex}]
Seeking a contradiction, suppose $\lvert \mathcal{X}_i\rvert \geq 4w + 1$. Since $\lvert F_i\rvert \leq w$, the set $\mathcal{X}_{\not \in i}$ of elements of $\mathcal{X}_i$ that are not incident with any edge in $F_i$ consists of at least $2w + 1$ vertices. Let $\xi: \mathcal{X}_{\not \in i} \to E(G)$ be the map associating to each vertex $z$ in $\mathcal{X}_{\not \in i}$ an edge $\xi(z)$ incident with $z$ in $G$ such that the last appearance of $\xi(z)$ is latest possible. 

Pick a vertex $z \in \mathcal{X}_{\not \in i}$ such that $\max \tau(\xi(z)) \leq \max \tau(\xi(z')) $ for any other $z' \in \mathcal{X}_{\not \in i}$. Since $\mathcal{X}_{\not \in i}$ contained at least $2w+1$ elements and since $\lvert F_{\max \tau(\xi(z))} \cup F_i\rvert \leq 2w$, there must be an element $y \in \mathcal{X}_{\not \in i}$ such that \[\max \tau(\xi(z)) < \min \tau(\xi(y)) < \max \tau(\xi(y)) < i.\]
By the definition of $\mathcal{X}_i$ and since $z \in \mathcal{X}_i$, there is some $(f,s,z) \in L_i$ and, by the previous Claim, there is a walk $W \in W(f,s,z)$. Notice that, since $W$ ends at the vertex $z$, it must be that the last time we `took a step' on $W$ was at a time at most $\max \tau(\xi(z))$; in particular this means that we did not move from $z$ at time $i$. But then, since $\max \tau(\xi(y)) < i$, $y$ never appears again after time $i -1$ and hence $W$ never traverses $\xi(y)$: this contradicts Property \ref{prop:IH1}.
\end{proof}
\end{proof}

As a corollary of Theorem \ref{thm:temp_Euler_dyn_prog}, we can leverage the reduction of Lemma  \ref{lemma:Euler_reduction} to deduce that $\starexp(k)$ is in $\fpt$ parameterized by the interval-membership-width. 

\begin{corollary}\label{corollary:starExp_algorithm}
There is an algorithm that, given a $\starexp(k)$ instance $(S_n,\tau)$, decides whether $(S_n,\tau)$ is explorable in time $\mathcal{O}(w^3 2^{3w} \Lambda)$ where $w = \imw(S_n, \tau)$ and $\Lambda$ is the lifetime of the input.
\end{corollary}
\begin{proof}
By Lemma  \ref{lemma:Euler_reduction}, we know that there is a polynomial-time reduction that maps any $\starexp(k)$ instance $(S_n, \tau)$ to a $\tempeuler(k-1)$-instance $(D_n, \sigma)$ such that
\begin{align*}
        \max_t \lvert \{&e \in E(D_n): \min(\sigma(e)) \le t \le \max(\sigma(e))\}\rvert \\ 
        &\leq 3 \max_t \lvert \{e \in E(S_n): \min(\tau(e)) \le t \le \max(\tau(e))\}\rvert. 
\end{align*}
In particular this implies that $\imw(D_n, \sigma) \leq 3 w$. Thus we can decide whether $(S_n, \tau)$ is explorable in time $\mathcal{O}(w^3 2^{3w} \Lambda)$ by applying the algorithm of Theorem \ref{thm:temp_Euler_dyn_prog} to $(D_n, \sigma)$.
\end{proof}

\section{Win-win approach to regularly spaced times}\label{sec_edge_expl:times_close_together}
In this section we will find necessary conditions for edge-explorability of temporal graphs with respect to their interval-membership-width. This will allow us to conclude that either we are given a no-instance or that the interval-membership-width is small (in which case we can employ our algorithmic results from the previous section).

We will apply this bidimensional approach to a variants of $\tempeuler(k)$ and $\starexp(k)$ in which we are given upper and lower bounds ($u$ and $\ell$ respectively) on the difference between any two consecutive times at any edge. Specifically we will show that $\starexp(k)$ is in $\fpt$ parameterized by $k$, $\ell$ and $u$ (Theorem \ref{theorem:even_explorability_algorithm}) and that $\tempeuler(k)$ is in $\fpt$ parameterized by $k$ and $u$ (Theorem \ref{thm:temp_Euler_fpt_by_k_and_u}). In other words, these results allow us to trade in the dependences on the interval-membership-width of Corollary \ref{corollary:starExp_algorithm} and Theorem \ref{thm:temp_Euler_dyn_prog} for a dependences on $k$, $\ell$, $u$ and $k$, $u$ respectively. 

We note that, for $\starexp$ instances, the closer $\ell$ and $u$ get, the more restricted the structure becomes to the point that the dependence on $\ell$ and $u$ in the running time of our algorithm vanishes when $\ell = u$. In particular this shows that the problem of determining the explorability of $\starexp(k)$-instances for which consecutive times at each edge are exactly $\lambda$ time-steps apart (for some $\lambda \in \mathbb{N}$) is in $\fpt$ parameterized solely by $k$ (Corollary \ref{corollary:even_star_algorithm}). This partially resolves an open problem of Akrida, Mertzios and Spirakis \cite{starexp} which asked to determine the complexity of exploring $\starexp(k)$-instances with evenly-spaced times.

Towards these results, we will first provide sufficient conditions for non-explorability of any $\starexp(k)$ instance (Lemma~\ref{lemma:clique-number}). These conditions will depend only on: (1) knowledge of the maximum and minimum differences between any two successive appearances of any edge, (2) the interval-membership-width and (3) the maximum number of appearances $k$ of any edge. 

\begin{lemma}\label{lemma:clique-number}
Let $(S_n, \tau)$ be a temporal star with at most $k$ times at any edge and such that every two consecutive times at any edge differ at least by $\ell$ and at most by $u$. If $(S_n, \tau)$ is explorable, then $\imw(S_n, \tau) \leq (2(k-1)u + 1)/(\ell + 1)$.
\end{lemma}
\begin{proof}
Let $\Lambda$ be the lifetime of $(S_n, \tau)$, let $(F_t)_{t \in [\Lambda]}$ be the interval-membership  sequence of $(S_n, \tau)$ and choose any $n \in [\Lambda]$ such that $\lvert F_n\rvert = \imw(S_n, \tau)$. Let $m$ and $M$ be respectively the earliest and latest times at which there are edges in $F_n$ which are active and chose representatives $e_m$ and $e_M$ in $F_n$ such that $m = \min \tau(e_m)$ and $M = \max \tau(e_M)$. 

Recall that visiting any edge $e$ in $S_n$ requires us to us pick two appearances (which differ by at least $\ell +1$ time-steps) of $e$ (one appearance to go along $e$ from the center of $S_n$ to the leaf and another appearance to return to the center of the star). Thus, whenever we specify how to visit an edge $e$ of $F_n$, we remove at least $\ell + 1$ time-steps from the available time-set $\{m, \ldots, M\}$ at which any other edge in $F_n$ can be visited (we need one time-step to travel to the leaf incident with $e$ and then - after $\ell$ time-steps -  we return to the center). Furthermore, since any exploration of $(S_n, \tau)$ must explore all of the edges in $F_n$, for $(S_n, \tau)$ to be explorable, we must have $\lvert F_n\rvert(\ell+1) \leq M - m + 1$. This concludes the proof since $\imw(S_n, \tau) = \lvert F_n\rvert$ and $M - m \leq \lvert \max \tau(e_M) - \min \tau(e_M)\rvert + \lvert \max \tau(e_m) - \min \tau(e_m)\rvert$ (since, by the definition of $F_n$, $n$ is in the intervals of any two elements of $F_n$) which is at most $2(k-1)u +1$ (since consecutive times at any edge differ by at most $u$).
\end{proof}

Notice that nearly-identical arguments yield the following slightly weaker result with respect to the $\tempeuler(k)$ problem.

\begin{lemma}\label{lemma:temp_euler_necessary_cond}
Let $(G, \tau)$ be a $\tempeuler(k)$ instance such that every two consecutive times at any edge differ at most by $u$. If $(G, \tau)$ is temporally Eulerian, then $\imw(G, \tau) \leq 2(k-1)u +1$.
\end{lemma}

The reason that the we can only bound $\imw(G, \tau)$ above by $2(k-1)u +1$ (rather than $(2(k-1)u +1)/(\ell + 1)$ as in the $\starexp(k)$ case of Lemma \ref{lemma:clique-number}) is that temporal Euler circuits only visit each edge once (so exploring each edge only removes exactly one available time).

Lemma \ref{lemma:clique-number} allows us to employ a `win-win' approach for $\starexp(k)$ when we know the maximum difference between consecutive times at any edge: either the considered instance does not satisfy the conditions of Lemma \ref{lemma:clique-number} (in which case we have a no-instance) or the interval-membership-width is small enough for us to usefully apply Corollary~\ref{corollary:starExp_algorithm}. These ideas allow us to conclude the following result. We point out that in the following Theorems~\ref{theorem:even_explorability_algorithm}, and \ref{thm:temp_Euler_fpt_by_k_and_u} Corollary~\ref{corollary:even_star_algorithm}, we can drop the factor $n$ in the running times if we assume that the relevant interval-membership sequences are given. 

\begin{theorem}\label{theorem:even_explorability_algorithm}
Let $(S_n, \tau)$ be a temporal star with at most $k$ times at any edge and such that every two consecutive times at any edge differ at least by $\ell$ and at most by $u$. There is an algorithm that decides whether $(S_n, \tau)$ is explorable in time $2^{\mathcal{O}(ku/\ell)} \Lambda$ where $\Lambda$ is the lifetime of the input.
\end{theorem}
\begin{proof}
Determine $\imw(S_n, \tau)$ (using the algorithm of Lemma~\ref{lemma:construction_algorithm_VERTEX_version}); if $\imw(S_n, \tau) > (2(k-1)u +1)/(\ell + 1)$, then $(S_n, \tau)$ is not explorable by Lemma \ref{lemma:clique-number}. Otherwise run the algorithm given in Corollary \ref{corollary:starExp_algorithm}. 
In this case, since $w := \imw(S_n, \tau) \leq (2(k-1)u +1)/(\ell + 1)$, we know that the algorithm of Corollary \ref{corollary:starExp_algorithm} will run on $(S_n,\tau)$ in time $2^{\mathcal{O}(ku/\ell)} \Lambda$. 
\end{proof}

Once again arguing by bidimensionality (this time using Lemma \ref{lemma:temp_euler_necessary_cond} and Theorem \ref{thm:temp_Euler_dyn_prog}) we can deduce the following fixed-parameter tractability result for $\tempeuler$.

\begin{theorem}\label{thm:temp_Euler_fpt_by_k_and_u}
Let $(G, \tau)$ be a $\tempeuler(k)$ instance such that every two consecutive times at any edge differ at most by $u$. There is an algorithm that decides whether $(G, \tau)$ is temporally Eulerian in time $2^{\mathcal{O}(ku)} \Lambda$ where $\Lambda$ is the lifetime of the input.
\end{theorem}

As a special case of Theorem \ref{theorem:even_explorability_algorithm} (i.e. the case where $\ell = u$) we resolve an open problem of Akrida, Mertzios and Spirakis \cite{starexp} which asked to determine the complexity of exploring $\starexp(k)$-instances with evenly-spaced times. In particular we show that the problem of deciding the explorability of such evenly-spaced $\starexp(k)$-instances is in $\fpt$ when parameterized by $k$.

\begin{corollary}\label{corollary:even_star_algorithm}
There is an algorithm that, given any $\starexp(k)$ instance $(S_n, \tau)$ with lifetime $\Lambda$ such that every two pairs of consecutive times assigned to any edge in $(S_n, \tau)$ differ by the same amount, decides whether $(S_n, \tau)$ is explorable in time $2^{\mathcal{O}(k)} \Lambda$.
\end{corollary}

\section{A vertex version of interval-membership-width}\label{sec:vertex_version}

It is natural to ask whether the parameter interval-membership-width also allows the design of $\fpt$ algorithms for problems involving vertex exploration or reachablity. 

Unfortunately, it turns out that restricting the interval-membership-width is not sufficient to guarantee tractability for certain natural problems. As an example, we show that $\minreachdelete$ (a problem concerning the deletion of edges to reduce the size of the largest reachability set in a temporal graph) remains hard even on instances with unit interval-membership-width (Theorem~\ref{thm:reachability_hardness}). These observations thus motivate the introduction of a vertex variant of our parameter -- called \emph{vertex}-interval-membership-width -- which proves to be more useful for such vertex-reachability problems. Indeed, in we show that parameterizing by the vertex-variant of interval-membership-width puts the aforementioned $\minreachdelete$ in $\fpt$ (Theorem~\ref{thm:reachability-fpt}). 

\subsection{Hardness of $\minreachdelete$}
A vertex $u$ is said to be \emph{temporally reachable} from $u$ in the temporal graph $(G,\tau)$ if there exists a temporal path from $v$ to $u$; every vertex is assumed to be temporally reachable from itself.  The \emph{temporal reachability set} of vertex $v$ in $(G,\tau)$, written $\reach_{G,\tau}(v)$, is then defined to be the set of vertices which are temporally reachable from $u$; the temporal reachability set of a set $S \subseteq V(G)$, written $\reach_{G,\tau}(S)$, is defined in the natural way to be $\bigcup_{v \in S} \reach_{G,\tau}(v)$.  The \emph{temporal reachability} of a set of vertices $S \subseteq V(G)$ is $\lvert \reach_{G,\tau}(S)\rvert$ (for a vertex $v$ in $(G,\tau)$ we write $\lvert \reach_{G,\tau}(v)\rvert$ rather than $\lvert \reach_{G,\tau}(\{v\})\rvert$).  We consider the following problem, introduced in \cite{Molter21DeleteDelay}. 

\begin{framed}
\noindent
\textbf{$\minreachdelete$}\\
\textbf{Input:} A temporal graph $(G,\tau)$, a set of source vertices $S$, and $k, h \in \mathbb{N}$.\\
\textbf{Question:} Is there a set $E'$ of time-edges, with $\lvert E'\rvert  \le k$, such that the temporal reachability of $S$ in $(G,\tau) \setminus E'$ is at most $h$?
\end{framed}

Note that this is a generalisation of the problem $\typesetproblem{TR-EdgeDeletion}$ introduced in \cite{TempEdgeDel}, where the set $S$ of sources is always taken to be equal to $V(G)$.  Here we adapt one of the arguments used to demonstrate intractability of $\typesetproblem{TR Edge Deletion}$ \cite[Theorem~3.1]{TempEdgeDel} to show that $\minreachdelete$ is para-NP-hard with respect to the interval-membership-width of the input graph.  

\begin{theorem}\label{thm:reachability_hardness}
$\minreachdelete$ is NP-hard, even if the input temporal graph has interval-membership-width one.
\end{theorem}
\begin{proof}
We will prove that the problem remains NP-hard even when the input temporal graph satisfies the following two properties:
\begin{enumerate}
\item every edge is active at exactly one time, and
\item no two edges are active simultaneously.
\end{enumerate}
These two conditions together immediately imply that the graph has interval-membership-width one.  Hardness of $\typesetproblem{TR Edge Deletion}$, and hence $\minreachdelete$, when every edge appears exactly once, was already demonstrated in \cite{TempEdgeDel}; however, in this construction, an unbounded number of edges is active at the same time (giving unbounded interval-membership-width).  Here we adapt the construction so that no two edges are active at the same time.  Note that, as each edge is active at exactly one time in our construction, we can use the terms edge and time-edge interchangeably in the proof.

As in \cite[Theorem~3.1]{TempEdgeDel}, the reduction is from the NP-hard problem \textsc{Clique}.  Let $(G,r)$ be the input to an instance of \textsc{Clique}, where $\lvert V(G)\rvert = \{v_1,\dots,v_n\}$ and $\lvert E(G)\rvert = \{e_1,\dots,e_m\}$.  We will construct an instance $((H,\tau),k,h)$ of \textsc{TR Edge Deletion} which is a yes-instance if and only if $(G,r)$ is a yes-instance for \textsc{Clique}.  As in the proof of \cite[Theorem~3.1]{TempEdgeDel}, we assume without loss of generality that $m > r + \binom{r}{2}$.

We begin by defining $H$.  The vertex set of $H$ is $V(H) = \{s\} \cup V(G) \cup E(G) \cup W$, where $W := \{w_{i,j}: 1 \le i \le n, 1 \le j \le m, v_i \in e_j\}$.  
The edge set is 
\begin{align*}
E(H) = & \{sv_i : 1 \le i \le n\} \\
	  & \cup \{v_iw_{i,j}: 1 \le i \le n, 1 \le j \le m, v_i \in e_j\} \\
	  & \cup \{w_{i,j}e_j : 1 \le i \le n, 1 \le j \le m, v_i \in e_j\} \\
	  & \cup \{sw_{i,j}: 1 \le i \le n, 1 \le j \le m, v_i \in e_j\}.
\end{align*}
We complete the construction of the temporal graph $(H,\tau)$ by setting
\[
\tau(e) = \begin{cases}
				i			& \mbox{if $e = sv_i$ for some $1 \le i \le n$}\\
				n + 2j - 1 	& \mbox{if $e = v_iw_{i,j}$ for some $1 \le i \le n$ and $1 \le j \le m$} \\
				n + 2j		& \mbox{if $e = w_{i,j}e_j$ for some $1 \le i \le n$ and $1 \le j \le m$} \\
				n + 2m + j	& \mbox{if $e = sw_{i,j}$ for some $1 \le i \le n$ and $1 \le j \le n$}.
		  \end{cases}
\]
It is immediate from the construction of $(H,\tau)$ that each edge is active at exactly one time, and that no two edges are active at the same time.  Finally, we set $S = \{s\}$, $k = r$ and $h = 1 + (n-r) + 2m + (m - \binom{r}{2})$.

Suppose first that $U \subseteq V(G)$ is a set of $r$ vertices that induces a clique in $G$.  Set $E' := \{sv: v \in U\}$ and write $(H',\tau')$ for the temporal graph obtained from $(H,\tau)$ by deleting all edges in $E'$.  We will argue that the temporal reachability of $S = \{s\}$ in $(H',\tau')$ is at most $h$, implying that $((H,\tau),S,k,h)$ is a yes-instance for $\minreachdelete$.  Note that $s$ reaches every vertex in $W$ along one-edge paths, but that no further vertices can be reached along paths starting with these edges as each such edge is active strictly later than any other edge incident with the endpoint in $W$.  It follows that every vertex reached by $s$ that does not belong to $W$ must be reached via a an element of $V(G) \setminus U$.  Thus we deduce that 
\[	
	\reach_{H',\tau'}(s) \subseteq \{s\} \cup W \cup V(G) \setminus U \cup \bigcup_{v \in V(G) \setminus U} \{e \in E(G): v \in e\}.
\]
In particular, we see that $s$ does not reach any vertex in $U$, or any element of $E(G)$ with both endpoints in $U$.  Since $U$ induces a clique in $G$, we see that $s$ fails to reach at least $r + \binom{r}{2}$ vertices in $(H',\tau')$ and hence $\lvert \reach_{H',\tau'}(s)\rvert \le 1 + n + m + 2m - r - \binom{r}{2} = h$, as required.

Conversely, suppose that there exists a set $E' \subseteq E(H)$ with $\lvert E'\rvert = r$ such that, if $(H',\tau')$ is the graph obtained from $(H,\tau)$ by deleting all edges in $E'$, we have that the reach of $s$ is $(H', \tau')$ is at most $\lvert \reach_{H',\tau'}(s)\rvert \le h$.  

Suppose first that $E' \cap \{sw: w \in W\} = \emptyset$; we will argue that in this case $s$ reaches at least $\lvert V(H)\rvert - \lvert E'\rvert - \binom{\lvert E'\rvert}{2}$ vertices in $(H',\tau')$ and that this lower bound can only be achieved if $G$ contains a clique on $r$ vertices.  

We begin by arguing that we may assume without loss of generality that every element of $E'$ is incident with $s$.  Suppose first that $v_iw_{i,j} \in E'$ for some $1 \le i \le n$ and $1 \le j \le m$.  The only vertices which are reached from $s$ along a temporal path using this edge are $e_j$ and elements of $W$ which are necessarily in the reachability set of $s$ since we are assuming we do not delete any edge of the form $sw$ with $w \in W$.  We can therefore replace $v_iw_{i,j}$ with $sv_i$ in $E'$ without increasing the number of vertices that are temporally reachable from $s$, since deleting $E'$ will still destroy the temporal path from $s$ to $e_j$ via $v_iw_{i,j}$.  Suppose now that $w_{i,j}e_j \in E'$ for some $1 \le i \le n$ and $1 \le j \le m$.  Again, the only vertex outside $W$ that is reachable from $s$ along a temporal path that includes this edge is $e_j$, and as before we can destroy this temporal path by instead deleting $sv_i$.  We therefore conclude that, provided that $E' \cap \{sw: w \in W\} = \emptyset$, it is possible to delete a subset of $\{sv_i: 1 \le i \le n\}$ of size $\lvert E'\rvert$ such that the reachability set of $s$ is a subset of $\reach_{H',\tau'}(s)$.  We will therefore assume from now on that $E' \subseteq \{sv_i: 1 \le i \le n\}$.

Set $U \subseteq V(G)$ to be the set of vertices in $V(G)$ that are incident with an edge in $E'$.  We claim that $U$ induces a clique in $G$.  To see this, note that $s$ reaches all of $V(G) \setminus U$, all of $W$, and every vertex in $E(G)$ that does not have both endpoints in $U$.  By assumption, we therefore have
\begin{align*}
	h &\ge \lvert \reach_{H',\tau'}(v)\rvert = 1 + \lvert V(G) \setminus U\rvert + \lvert W\rvert + \left\lvert \bigcup_{v \in V(G) \setminus U} \{e \in E(G): v \in e\}\right\rvert \\
	&= 1 + (n-r) + 2m + m - \{xy \in E(G): x,y \in U\}.
\end{align*}
By definition of $h$, it follows that $\lvert \{xy \in E(G): x,y \in U\}\rvert \ge \binom{r}{2}$; this holds if and only if $U$ induces a clique in $G$, in which case we have equality.  It therefore follows, as claimed, that $s$ reaches at least $\lvert V(H)\rvert - \lvert E'\rvert - \binom{\lvert E'\rvert}{2}$ vertices in $(H',\tau')$, and that this lower bound can only be achieved if $G$ contains a clique on $r$ vertices.  

To complete the proof, we argue that the remaining case $E' \cap \{sw: w \in W\} \neq \emptyset$ cannot occur.  Suppose that $\lvert E' \cap \{sw: w \in W\}\rvert = t > 0$.  Note that the deletion of an edge $sw$ with $w \in W$ can at most remove $w$ from the reachability set of $s$, since the edge $sw$ is active strictly later than any other edge incident with $w$ and so cannot be part of a longer temporal path starting at $s$.  If we set $E'' = E' \setminus \{sw: w \in W\}$, this observation combined with the reasoning above tells us that 
\begin{align*}	
	\lvert \reach_{H',\tau'}(s)\rvert &\ge \lvert V(H)\rvert - \lvert E''\rvert - \binom{\lvert E''\rvert}{2} - t \\
	&= \lvert V(H)\rvert - (r-t) - \binom{\lvert E''\rvert}{2} - t \\
	&> \lvert V(H)\rvert - r - \binom{r}{2} = h,
\end{align*}
contradicting our choice of $E'$.
\end{proof}

\subsection{\emph{Vertex}-interval-membership-width}
We now introduce the vertex version of interval-membership-width, which captures the maximum number of vertices incident both an edge active before time $t$ and an edge active after time $t$, taken over all times $t$ in the lifetime of the graph.

\begin{definition}
Let $(G,\tau)$ be a temporal graph with lifetime $\Lambda$. The \emph{vertex interval-membership  sequence} of $(G,\tau)$ is the sequence $(F_t)_{t \in [\Lambda]}$ of vertex-subsets of $G$ (called \emph{bags}) where each $F_t$ is defined as follows
\[
    F_t := \{v \in V(G): \exists i \le t \le j \text{ and } u,w \in V(G) \text{ such that } i \in \tau(uv) \text{ and } j \in \tau(wv)\}
\]
(note that $u$ and $w$ need not be distinct). The \emph{vertex-interval-membership-width} of $(G,\tau)$ -- denoted $\vimw(G,\tau)$ -- is the maximum cardinality attained by any bag in the vertex interval-membership  sequence of $(G,\tau)$ (i.e. $\vimw(G,\tau) := \max_{t \in [\Lambda]} \lvert F_t\rvert$).
\end{definition}

It is clear that, for any temporal graph $(G,\tau)$, $\vimw(G,\tau) \ge 2 \imw(G,\tau)$: if an edge $e$ is active at times before and after $t$, and so belongs to the set for time $t$ in the interval-membership  sequence, it follows that both endpoints of $e$ must belong to the set for time $t$ in the vertex interval-membership  sequence.  However, the difference between the two parameters can be arbitrarily large.  Consider, for example, a disjoint union of two-edge paths $P_1,\dots,P_k$, where the edges of $P_i$ are active at times $i$ and $i+k$ respectively: since every edge appears at a single unique time, the interval-membership-width is only one, but the midpoint of every path is incident with an edge appearing at a time at most $k$ and a time strictly greater than $k$, giving vertex-interval-membership-width at least $k$. 

\begin{lemma}\label{lemma:construction_algorithm_VERTEX_version}
If $(G, \tau)$ is a temporal graph where each edge is active at-least once, then one can compute the vertex-interval-membership sequence of $(G, \tau)$ in time $\mathcal{O}(\vimw(G, \tau)\Lambda)$.
\end{lemma}
\begin{proof}
Initialize a sequence $(W_t)_{t \in [\Lambda]}$ of empty bags; then proceed as follows: 
\begin{enumerate}[label=\textbf{V\arabic*}]
    \item For each edge $e$ of $G$, find the pair $(a_e, A_e) = (\min \tau(e), \max \tau(e))$; \label{step:vertex_1}
    \item for each vertex $x$ of $G$, \label{step:vertex_2}
        \begin{enumerate}
            \item let $(b_x, B_x) = (\min_{e \ni x} a_e, \max _{e \ni x} A_e)$
            \item for all times $t$ with $b_x \leq t \leq B_x$, add $x$ to the bag $W_t$.
        \end{enumerate}
\end{enumerate}
Correctness follows trivially from the definition of the vertex-interval-membership sequence, so now we consider the running time. Denoting the interval-membership sequence of $(G, \tau)$ by $(F_t)_{t \in [\Lambda]}$, then, since $2 \imw(G, \tau) \leq \vimw(G, \tau)$ (as we observed earlier), Step~\ref{step:vertex_1} takes time \[\mathcal{O}\Bigl(\sum_{e \in E(G)}\lvert \tau(e)\rvert\Bigr) = \mathcal{O}\Bigl(\sum_{t \in [\Lambda]}\lvert F_t\rvert\Bigr) = \mathcal{O}(\imw(G, \tau)\Lambda ) = \mathcal{O}(\vimw(G, \tau)\Lambda).\]
Step~\ref{step:vertex_2} takes time 
\begin{align*}
    \mathcal{O}\Bigl(&\sum_{x \in V(G)} (d_G(x) + \lvert \{t \in [\Lambda]: x \in W_t\}\rvert)\Bigr) = \\
    &= \mathcal{O}(\lvert E(G)\rvert) + \mathcal{O}\Bigl(\sum_{x \in V(G)} \lvert \{t: x \in W_t\}\rvert)\Bigr) \\
    &= \mathcal{O}\Bigl(\sum_{t \in [\Lambda]} \lvert F_t\rvert)\Bigr) + \mathcal{O}\Bigl(\sum_{x \in V(G)} \lvert \{t: x \in W_t\}\rvert)\Bigr) &(\forall e\in E(G),\:\lvert \tau(e)\rvert\geq1)\\
    &= \mathcal{O}\Bigl(\sum_{t \in [\Lambda]} \lvert F_t\rvert)\Bigr) + \mathcal{O}\Bigl(\sum_{t \in [\Lambda]}\lvert W_t\rvert)\Bigr) \\
    &= \mathcal{O}(\imw(G, \tau)\Lambda ) + \mathcal{O}(\vimw(G, \tau)\Lambda) \\
    &= \mathcal{O}(\vimw(G, \tau)\Lambda) &(\imw(G, \tau) \leq \vimw(G, \tau)).
\end{align*}
Thus the whole algorithm runs in $\mathcal{O}(\vimw(G, \tau)\Lambda)$ time, as desired.
\end{proof}

We now illustrate the greater algorithmic power of vertex-interval-membership-width by showing that $\minreachdelete$ admits an $\fpt$ algorithm with respect to this larger parameter.

\begin{theorem}\label{thm:reachability-fpt}
There is an algorithm that decides whether any temporal graph $(G,\tau)$ with lifetime $\Lambda$ is a yes-instance of $\minreachdelete$ in time $\mathcal{O}(w^2 h 2^{w^2} \Lambda)$, where $w = \vimw(G,\tau)$ is the vertex-interval-membership-width of $(G,\tau)$.
\end{theorem}
\begin{proof}
We proceed by a dynamic programming argument similar to that of Theorem~\ref{thm:temp_Euler_dyn_prog}.  Let $((G,\tau),S,k,h)$ be the input to an instance of $\minreachdelete$ and let $(F_t)_{t \in [\Lambda]}$ be the vertex interval-membership  sequence of $(G,\tau)$, where without loss of generality that $F_1$ is non-empty.  

For each $i \in [\Lambda]$, we compute a set $L_i \subseteq \{0,\dots,h\} \times F_i^{\{0,1\}}$ consisting of pairs of the form $(r,f)$ where $r$ is an integer between $0$ and $h$ and $f$ is a function from $F_i$ to $\{0,1\}$.  We define $L_i$ to be the set of such pairs $(r,f)$ such that there exists a set $E_i'$ of time-edges with the following properties:
\begin{itemize}
\item $\lvert E_i'\rvert \le k$, and
\item if $R$ denotes the set of vertices reachable from $S$ in $(G,\tau) \setminus E_i$ no later than time $i$, then
\begin{itemize}
\item $\lvert R\rvert \le r$, and
\item for every $v \in F_i$, $v \in R$ if and only if $f(v) = 1$.
\end{itemize} 
\end{itemize}
We say that such a set $E_i'$ of time-edges \emph{witnesses} the element $(r,f)$ for $L_i$.  Note that we may assume without loss of generality that every time-edge in $E_i'$ appears at a time less than or equal to $i$.  For each $(r,f) \in L_i$, we set the \emph{cost} of $(r,f)$, written $\cost_i(r,f)$, to be the minimum cardinality of any set of time-edges witnessing $(r,f)$ for $L_i$.  It is clear from these definitions that we have a yes-instance for $\minreachdelete$ if and only if $L_{\Lambda} \neq \emptyset$, and that in this case the minimum number of edges we need to delete is equal to $\min_{(r,f) \in L_{\Lambda}}\cost_{\Lambda}(r,f)$.

It therefore remains to demonstrate that we can compute all sets $L_i$ within the stated time; in fact, we shall also compute the cost function for each $L_i$ as we will make use of this in computing $L_{i+1}$.  Note first that, for all $i \in [\Lambda]$, $\lvert L_i\rvert \le h2^w$.  Moreover, recall from the definition of the vertex interval-membership  sequence that both endpoints of every edge active at time $i$ belong to $F_i$, and so the number of edges active at time $i$ is at most $\binom{\lvert F_i\rvert}{2} \le \lvert F_i\rvert^2$.

We begin by considering $L_1$, and let $E_1$ be the set of time-edges appearing at time one.  Since $\lvert E_1\rvert \le \lvert F_1\rvert^2 \le w^2$, we can consider all $2^{w^2}$ possibilities for a subset $E_1' \subseteq E_1$ to delete; for each such subset of cardinality at most $k$, it is clear that we can compute in time $\mathcal{O}(\lvert E_1\rvert') = \mathcal{O}(w^2)$ the corresponding pair $(r,f)$ and record $\lvert E_1'\rvert$ as an upper bound for $\cost_1(r,f)$, since every vertex outside $S$ that is reachable from $S$ by time $1$ must belong to $F_1$.  We can therefore compute $L_1$ in time $\mathcal{O}(w^22^{w^2})$.

Now supposing that we have computed $L_i$ and the function $\cost_i : L_i \rightarrow \{0,\dots,k\}$, we will explain how to compute $L_{i+1}$ and the function $\cost_{i+1} : L_{i+1} \rightarrow \{0,\dots,k\}$.  Observe that every vertex reached from $S$ by time $i+1$ that was not already reached by time $i$ must belong to $F_{i+1}$: such a vertex must be reached along an edge appearing at time $i+1$ and so is incident with such an edge.  Moreover, any such vertex $v$ is reached via an element of $F_i \cup S$: either $v$ is reached directly from $S$ via an edge active at time $i+1$, or $v$ is reached from another vertex $u$, which was reached by time $i$; since $u$ is therefore incident with an edge appearing at a time at most $i$ and an edge active at time $i+1$, we conclude that $u \in F_i \cap F_{i+1}$.  

It follows that the set of vertices reachable from $S$ by time $i+1$ can be computed from the set of vertices in $F_i \cap F_{i+1}$ reachable from $S$ by time $i$ together with the set of edges active at time $i+1$.  Writing $E_{i+1}$ for the number of edges active at time $i+1$, we have $\lvert E_{i+1}\rvert \le \binom{w}{2}$, and so we can consider each of the $2^{\binom{w}{2}}$ possibilities for which edges in this set to delete.  For each subset $E_{i+1}' \subseteq E_{i+1}$ and for every state $(r,f) \in L_i$, set 
\begin{align*}
	R_{f,E_{i+1}'} := &\{v \in F_{i+1} \colon \exists s \in S \text{ with } sv \in E_{i+1} \setminus E_{i+1}'\} \\
	&\cup \{v \in F_{i+1} \colon \exists u \in F_i \text{ with } f(u) = 1 \text{ and } uv \in E_{i+1} \setminus E_{i+1}'\}.
\end{align*}
By the previous reasoning, it is clear that $R_{f, E_{i+1}'}$ is precisely the set of vertices reached at time exactly $i+1$ under the assumptions that 
\begin{enumerate}
\item $f^{-1}(1)$ is the set of vertices in $F_i$ reachable by time $i$, and
\item $E_{i+1}'$ is the set of time-edges deleted at time $i+1$.
\end{enumerate} 
For each subset $E_{i+1}' \subseteq E_{i+1}$ and for every state $(r,f) \in L_i$ such that 
\begin{enumerate}
\item $\cost_i(r,f) + \lvert E_{i+1}'\rvert \le k$, and
\item $r + \lvert R_{f,E_{i+1}'}\rvert \le h$,
\end{enumerate} we therefore add to $L_{i+1}$ the state $(r',f')$, where $r' = r + \lvert R_{f,E_{i+1}'}\rvert$ and, for all $v \in F_{i+1}$,
\[
	f'(v) = \begin{cases}
				1	& \mbox{if $v \in R_{f,E_{i+1}'}$} \\
				0	& \mbox{otherwise.}
			\end{cases}
\]	
We also record $\cost_i(r,f) + \lvert E_{i+1}'\rvert$ as an upper bound for $\cost_{i+1}(r',f')$.  After iterating through all possibilities for $E_{i+1}'$ and $(r,f)$, the true value of $\cost_{i+1}(r',f')$ is the least upper bound we have recorded for this state in $L_{i+1}$.

It remains only to bound the time needed to compute $L_{i+1}$ and the associated cost function as described.  To do this, we consider each pair consisting of a subset $E_{i+1}'$ of time-edges to delete, for which there are $\mathcal{O}(2^{\binom{w}{2}})$ possibilities, together with a state $(r,f) \in L_i$, for which there are $\mathcal{O}(h2^w)$.  The total number of pairs we consider is therefore $\mathcal{O}(h2^{w^2})$.  For each such pair, we can compute the corresponding state in $L_{i+1}$ and the upper bound on the associated cost by examining each of the edges in $E_{i+1} \setminus E_{i+1}'$, of which there are $\mathcal{O}(w^2)$.  We therefore compute $L_{i+1}$ and the function $\cost_{i+1}$ in time $\mathcal{O}(w^2h2^{w^2})$.

Summing over all sets $L_i$ for $1 \le i \le \Lambda$ we see that the time needed to compute all states and associated cost functions is $\mathcal{O}(w^2 h 2^{w^2} \Lambda)$, as required.
\end{proof}

\section{Discussion}\label{sec_edge_expl:conclusion}
In this paper we introduced a natural temporal analogue of Eulerian circuits and proved that, in contrast to the static case, $\tempeuler(k)$ is $\np$-complete for all $k\geq 3$. In fact we showed that, for $k \geq 3$, the problem remains hard even when the underlying static graph has path-width $2$, feedback vertex number $1$ or vertex cover number $2$ (Section~\ref{sec_edge_expl:hardness}).

Independently and simultaneously to our work here, Marino and Silva~\cite{marino2021k} showed that  $\tempeuler(k)$ is $\np$-complete for all $k\geq2$ (thus resolving the $k=2$ case which we had left open). 

While proving our hardness results for $\tempeuler(k \geq 3)$, we also resolved an open problem of Akrida, Mertzios and Spirakis \cite{starexp} by showing that $\starexp(k)$ is $\np$-complete for all $k \geq 4$. This result yields a complete complexity dichotomy with respect to $k$ when combined with Akrida, Mertzios and Spirakis' results \cite{starexp}.

Our hardness results rule out $\fpt$ algorithms for $\tempeuler(k)$ and $\starexp(k)$ with respect to many standard parameters describing the structure of the underlying graph (such as path-width, feedback vertex number and vertex-cover number). Motivated by these resutls, we introduced a new width measure which captures structural information that is purely temporal; we call this the \emph{interval-membership-width}. In contrast to our hardness results, we showed that $\tempeuler(k)$ and $\starexp(k)$ can be solved in times $\mathcal{O}(w^3 2^w \Lambda)$ and $\mathcal{O}(w^3 2^{3w} \Lambda)$ respectively where $w$ is our new parameter and $\Lambda$ is the lifetime of the input. 

Our fixed-parameter-tractability results parameterized by interval-membership-width can also be leveraged via a win-win approach to obtain tractability results for both $\tempeuler(k)$ and $\starexp(k)$ parameterized solely by $k$ and the minimum and maximum differences between any two successive times in a time-set of any edge. These resutls allow us to partially resolve another open problem of Akrida, Mertzios and Spirakis concerning the complexity of $\starexp(k)$: we showed that it can be solved in time $2^{\mathcal{O}(k)}\Lambda$ when the input has evenly spaces appearances of each edge and lifetime $\Lambda$. We note, however, that it remains an open problem to determine the complexity of the evenly-spaced $\starexp(k)$ problem when $k$ is unbounded. 

Given the success of parameterizations by interval-membership-width when it comes to temporal edge-exploration problems, it is natural to ask whether such parametrizations can also yield $\fpt$  algorithms for problems involving vertex exploration or reachablity. We showed that the vertex-reachability problem $\minreachdelete$ remains hard even on temporal graphs which have unit interval-membership-width. This motivated us to introduce a `vertex-variant' of our measure called \emph{vertex-interval-membership-width}. This parameter is bounded below by interval-membership-width, but the difference between the two can be arbitrarily large. Parameterizing by this larger parameter, puts $\minreachdelete$ in $\fpt$ and demonstrates the greater algorithmic power of vertex-interval-membership-width. 

Finally we point out that all of our hardness reductions hold also for the case of non-strict temporal walks and, with slightly more work, even our tractability results can be seen to hold for the non-strict case.

\paragraph{Acknowledgements:} The authors would like to thank Samuel Hand and John Sylvester for spotting slight inaccuracies in the preliminary version of this article and the anonymous reviewers for their helpful comments and suggestions.

\bibliography{biblio}
\bibliographystyle{plainurl}

\end{document}